\documentclass[11pt]{article}

\usepackage[cmex10]{amsmath}
\usepackage{amssymb,amsfonts,textcomp,amsthm}
\usepackage{array}
\usepackage{delarray}
\usepackage{latexsym}
\usepackage[pdftex]{color, graphicx}
\usepackage{subfigure}
\usepackage{mathptmx}
\usepackage{times}
\usepackage{paralist}
\usepackage{tabularx}
\usepackage{multirow}
\usepackage{graphicx}
\usepackage{url}
\usepackage{rotating}
\usepackage{cite}
\usepackage[ruled, vlined, noend, linesnumbered, commentsnumbered]{algorithm2e}
\usepackage[font=footnotesize,labelfont=bf]{caption}

\newcommand{\ltot}{{{l_{\mathit{tot}}}}}
\newcommand{\impr}{{{\mathit{impr}}}}
\newcommand{\lmax}{{{l_{\mathit{max}}}}}
\newcommand{\lmaxi}{{{l_{\mathit{max}, i}}}}

\newtheorem{theorem}{Theorem}

\newtheorem{proposition}[theorem]{Proposition}
\newtheorem{corollary}[theorem]{Corollary}

\newtheorem{definition}{Definition}
\newtheorem{lemma}[theorem]{Lemma}

\newenvironment{proof-sketch}{\noindent \textit{Sketch of Proof:}}{$\Box$}

\begin{document}

\title{We Are Impatient: Algorithms for Geographically Distributed Load Balancing with (Almost) Arbitrary Load Functions}

\author{Piotr Skowron \\
   Institute of Informatics \\
   University of Warsaw \\
   Email: p.skowron@mimuw.edu.pl \\
   \and
   Krzysztof Rzadca \\
   Institute of Informatics \\
   University of Warsaw \\
   Email: krzadca@mimuw.edu.pl}

\date{}

\maketitle

\begin{abstract}

% problem: key topic or problem of the paper
%We consider the problem of balancing load in a distributed system. 
In geographically-distributed systems, 
%In contrast to the common approach to load balancing in this paper we assume that the 
communication latencies are non-negligible.
The perceived processing time of a request is thus composed of the time needed to route the request to the server and the true processing time.
 % and thus affect the perceived processing time.
Once a request reaches a target server, the processing time depends on the total load of that server; this dependency is described by a load function.

% method: state your main approach to solving the problem
We consider a broad class of load functions; we just require that they are convex and two times differentiable. In particular our model can be applied to heterogeneous systems in which every server has a different load function. This approach allows us not only to generalize results for queuing theory and for batches of requests, but also to use empirically-derived load functions, measured in a system under stress-testing.

% results
The optimal assignment of requests to servers is communication-balanced, i.e. for any pair of non perfectly-balanced servers, the reduction of processing time resulting from moving a single request from the overloaded to underloaded server is smaller than the additional communication latency. 

We present a centralized and a decentralized algorithm for optimal load balancing.
% Both algorithms can be applied for almost any load function. 
% For every load function 
We prove bounds on the algorithms' convergence. To the best of our knowledge these bounds were not known even for the special cases studied previously (queuing theory and batches of requests).
% Apart from standard parameters these bounds take into account the bounds on the derivatives of the load functions. 
Both algorithms are any-time algorithms. 
In the decentralized algorithm, each server balances the load with a randomly chosen peer. Such algorithm is very robust to failures. We prove that the decentralized algorithm performs locally-optimal steps. 
% For both algorithms we are able to estimate the error during their execution.

% conclusions
Our work extends the currently known results by considering a broad class of load functions and by establishing theoretical bounds on the algorithms' convergence. These results are especially applicable for servers whose characteristics under load cannot be described by a standard mathematical models.
% ; our methods load in which we need to derive load functions experimentally.

%% problem: key topic or problem of the paper
% balancing the load in a distributed systems; communication latency between servers is non-negligible
% the time needed to process a request depends on the total load of the server
%% method: state your main approach to solving the problem
% very broad class of processing time function; just assume they are differentiable that allows us to generalize previous work on queueing theory, but also use empirical load functions (measured in a system under stress-testing)
%% results
% an optimal solution is comm-balanced (balanced taking into account communication latency); i.e. for any pair of non perfectly-balanced servers, the gain we'd have by moving a single request from the overloaded to underloaded is smaller than the communication latency
% a centralized load-balancing algorithm: add properties
% a decentralized optimization algorithm: each server balances the load with a peer chosen randomly. We prove that the algorithm converges; and compute the complexity
% conclusions
% in some systems, response time cannot be modelled well only using queueing theory; our theory can be applied in just these settings
\end{abstract}

\newpage

\section{Introduction}\label{sec::introduction}

% subject of the paper
We are impatient. 
An ``immediate'' reaction must take less than 100~ms~\cite{card1983psychology};
a Google user is less willing to continue searching if the result page is slowed by just 100-400~ms~\cite{google2009:speed-matters};
and a web page loading faster by just 250~ms attracts more users than the competitor's~\cite{lohr2012impatient}.
Few of us are thus willing to accept the 100-200ms Europe-US round-trip time; even fewer, 300-400ms Europe-Asia.
Internet companies targeting a global audience must thus serve it locally. 
Google builds data centers all over the world; a company that doesn't have Google scale uses a generic content delivery network (CDN)~\cite{pallis2006content, freedman2010experiences}, such as Akamai~\cite{Nygren:2010:ANP:1842733.1842736, akamai2, draftingBehindAkamai}; or spreads its content on multiple Amazon's Web Service regions.

A geographically-distributed system is an abstract model of world-spanning networks. It is a network of interconnected servers processing requests. The system considers both communication (request routing) and computation (request handling). E.g., apart from the communication latencies, a CDN handling complex content can no longer ignore the load imposed by requests on servers. As another example consider computational clouds, which are often distributed across multiple physical locations, thus must consider the network latency in addition to servers' processing times.

% define problem we're addressing && why this problem is important
Normally, each server handles only the requests issued by local users. For instance, a CDN node responds to queries incoming from the sub-network it is directly connected to (e.g., DNS redirections in Akamai~\cite{draftingBehindAkamai, akamai2, akamaiPatent}). 
However, load varies considerably. 
Typically, a service is more popular during the day than during the night (the daily usage cycle).
Load also spikes during historic events, ranging from football finals to natural disasters.
If a local server is overloaded, some requests might be handled faster on a remote, non-overloaded server.
The users will not notice the redirection if the remote server is ``close'' (the communication latency is small);
but if the remote server is on another continent, the round-trip time may dominate the response time.

% scope: list assumptions
In this paper we address the problem of balancing servers' load taking into account the communication latency. 
We model the response time of a single server by a \emph{load function}, i.e., a function that for a given load on a server (the number of requests handled by a server) returns the average processing time of requests. In particular, we continue the work of Liu~et~al.~\cite{Liu:2011:GGL:1993744.1993767}, who showed the convergence of the algorithms for the particular load function that describes requests' handling time in the queuing model~\cite{Gross:2008:FQT:1972549}.
We use a broad class of functions that are continuous, convex and twice-differentiable (Section~\ref{sec:mathematical-model}), 
which allows us to model not only queuing theory-based systems,
but also a particular application with the response time measured empirically in a stress-test.

We assume that the servers are connected by links with high bandwidth. 
Although some models (e.g., routing games \cite{routingGames}) consider limited bandwidth, our aim is to model servers connected by a dense network (such as the Internet), in which there are multiple cost-comparable routing paths between the servers.
The communication time is thus dominated by the latency: 
a request is sent over a long distance with a finite speed. We assume that the latencies are known, as monitoring pairwise latencies is a well-studied problem~\cite{Szymaniak04scalablecooperative, Chan-TinH11}; if the latencies change due to, e.g., network problems, our optimization algorithms can be run again.
On each link, the latency is constant, i.e., it does not vary with the number of sent requests~\cite{Skowron:2013:NDL:2510648.2510769}. This assumption is consistent with the previous works on geographically load balancing~\cite{Liu:2011:GGL:1993744.1993767, Skowron:2013:NDL:2510648.2510769, Cardellini00geographicload, Penmatsa:2006:CLB:1898953.1899089, Grosu:2008:CLB:1455689.1455695, Aote:2009:GMD:1523103.1523153, Grosu:2002:AMD:792762.793282}.

Individual requests are small; rather than an hour-long batch job, a request models, e.g., a single web page hit. Such assumption is often used~\cite{Penmatsa:2006:CLB:1898953.1899089, Liu:2011:GGL:1993744.1993767, Skowron:2013:NDL:2510648.2510769, Grosu:2008:CLB:1455689.1455695, gallet09divisibleload, beaumont2005scheduling, veeravalli2002efficient, drozdowski2008scheduling}. In particular the continuous allocation of requests to servers in our model is analogous to the divisible load model with constant-cost communication (a special case of the affine cost model~\cite{beaumont2005scheduling}) and multiple sources (multiple loads to be handled,~\cite{veeravalli2002efficient, drozdowski2008scheduling}).

%related work
The problem of load balancing in geographically distributed systems has been already addressed. Liu~et~al.~\cite{Liu:2011:GGL:1993744.1993767} shows the convergence of the algorithms for a particular load function from the queuing theory. 
Cardellini~et~al.~\cite{Cardellini00geographicload} analyzes the underlying system and network model. Colajanni~et~al.~\cite{Colajanni98dynamicload} presents experimental evaluation of a round-robin-based algorithm for a similar problem.
Minimizing the cost of energy due to the load balancing in geographically distributed systems is a similar problem considered in the literature~\cite{Liu:2013:DCD:2494232.2465740, Liu:2011:GLB:2160803.2160862, Lin:2012:OAG:2410145.2410785}. A different load function is used by Skowron~and~Rzadca~\cite{Skowron:2013:NDL:2510648.2510769} to model the flow time of batch requests.

Some papers analyze the game-theoretical aspects of load balancing in geographically distributed systems~\cite{Penmatsa:2006:CLB:1898953.1899089, Grosu:2008:CLB:1455689.1455695, Aote:2009:GMD:1523103.1523153, Grosu:2002:AMD:792762.793282, Adolphs:2012:DSL:2332432.2332460}. These works use a similar model, but focus on capturing the economic relation between the participating entities.

The majority of the literature ignores the communication costs.
Our distributed algorithm is the extension of the diffusive load balancing~\cite{conf/ipps/AdolphsB12, Ackermann:2009:DAQ:1583991.1584046, Berenbrink:2011:DSL:2133036.2133152}; it incorporates communication latencies into the classical diffusive load balancing algorithms.

Additionally to the problem of effective load balancing we can optimize the choice of the locations for the servers~\cite{staticReplicaPlacement1, staticReplicaPlacement2, journals/cj/JiaLHD01}. The generic formulation of the placement problem, facility location problem~\cite{Chudak:2005:IAA:1047770.1047776} and k-median problem~\cite{Jain:2001:AAM:375827.375845} have been extensively studied in the literature.

%We also assume that the processing time of all requests assigned to a particular server is the same (and depends on the load function). We do not schedule individual requests, but rather address the steady-state problem. 

% approach to the problem --- main results with links
\medskip
\textbf{The contributions of this paper are the following.}
\begin{inparaenum}[(i)] 
\item We construct a centralized load-balancing algorithm that optimizes the response time up to a given (arbitrary small) distance to the optimal solution (Section~\ref{sec:approximate-centralized}). 
The algorithm is polynomial in the total load of the system and in the upper bounds of the derivatives of the load function.
\item We show a decentralized load-balancing algorithm (Section~\ref{sec:distr-algor}) in which pairs of servers balance their loads. We prove that the algorithm is optimal (there is no better algorithm that uses only a single pair of servers at each step). We also bound the number of pair-wise exchanges required for convergence.
\item We do not use a particular load function; instead, we only require the load function to be continuous and twice-differentiable (Section~\ref{sec:mathematical-model}); 
thus we are able to model empirical response times of a particular application on a particular machine,
but also to generalize (Section~\ref{sec:specialCases}) Liu~et~al.~\cite{Liu:2011:GGL:1993744.1993767}'s results on queuing model and the results on flow time of batch requests~\cite{Skowron:2013:NDL:2510648.2510769}.
\end{inparaenum}

%significance

%The considered problem is important. A system using our algorithms will have smaller response times and thus better user experience.
%Additionally, when the operator is planning to upgrade the infrastructure, our results can help in deciding where to install the additional servers.

Our algorithms are suitable for real applications. Both are any-time algorithms which means that we can stop them at any time and get a complete, yet suboptimal, solution.
%Alternatively, if necessary, we can compute the solution longer, getting a better quality, if necessary. 
Furthermore, the distributed algorithm is particularly suitable for distributed systems. It performs only pairwise optimizations (only two servers need to be available to perform a single optimization phase), which means that it is highly resilient to failures. It is also very simple and does not require additional complex protocols.

In this paper we present the theoretical bounds, but we believe that the algorithms will have even better convergence in practice. The experiments have already confirmed this intuition for the case of distributed algorithm used for a batch model~\cite{Skowron:2013:NDL:2510648.2510769}. The experimental evaluation for other load functions is the subject of our future work.

\section{Preliminaries}

In this section we first describe the mathematical model, and next we argue that our model is highly applicable; in particular it generalizes the two problems considered in the literature.

\subsection{Model}\label{sec:mathematical-model}

\noindent
\textbf{Servers, requests, relay fractions, current loads.}\quad
The system consists of a set of $m$ \emph{servers} (processors) connected to the Internet. The $i$-th server has its \emph{local (own) load} of size $n_{i}$ 
%(the size of the load can be thought of as processing requirement)
consisting of small \emph{requests}. The local load can be the current number of requests, the average number of requests, or the rate of incoming requests in the queuing model.

The server can relay a part of its load to the other servers. We use a fractional model in which a \emph{relay fraction} $\rho_{ij}$ denotes the fraction of the $i$-th server's load that is sent (relayed) to the $j$-th server ($\forall_{i,j} \: \rho_{ij} \geq 0$ and $\forall _{i} \: \sum_{j = 1}^{j = m}\rho_{ij} = 1$). Consequently, $\rho_{ii}$ is the part of the $i$-th load that is kept on the $i$-th server. We consider two models. In the \emph{single-hop model} the request can be sent over the network only once. In the \emph{multiple-hop} model the requests can be routed between servers multiple times\footnote{We point the analogy between the multiple-hop model and the Markov chain with the servers corresponding to states and relay fractions $\rho_{ij}$ corresponding to the probabilities of changing states.}. Let $r_{ij}$ denote the size of the load that is sent from the server $i$ to the server $j$. In the single-hop model the requests transferred from $i$ to $j$ come only from the local load of the server $i$, thus:
\begin{align}
r_{ij} = \rho_{ij}n_i \textrm{.}
\end{align}
In the multiple-hop model the requests come both from the local load of the server $i$ and from the loads of other servers that relay their requests to $i$, thus $r_{ij}$ is a solution of:
\begin{align}
r_{ij} = \rho_{ij}\left(n_i + \sum_{k \neq i}r_{ki}\right) \textrm{.}
\end{align}
The \emph{(current) load} of the server $i$ is the size of the load sent to $i$ by all other servers, including $i$ itself:
%\begin{align}
$l_i = \sum_{j = 1}^{m}r_{ji} \textrm{.}$
%\end{align}
\medskip

\noindent
\textbf{Load functions.}\quad
Let $f_i$ be a \emph{load function} describing the average request's processing time on a server $i$ as a function of $i$'s load $l_i$ (e.g.: if there are $l_i=10$ requests and $f_i(10)=7$, then on average it takes $7$ time units to process each request). We assume $f_i$ is known from a model or experimental evaluation; but each server can have a different characteristics $f_i$ (heterogeneous servers).
The total processing time of the requests on a server $i$ is equal to $h_i(l_i) = l_{i} f_{i}(l_i)$ (e.g., it takes $70$ time units to process all requests). 
In most of our results we use $f_i$ instead of $h_i$ to be consistent with~\cite{Liu:2011:GGL:1993744.1993767}.

%We recall that $l_i$ is the current load on the $i$-th server

Instead of using a certain load function, we derive all our results for a broad class of load functions (see Section~\ref{sec:specialCases} on how to map existing results to our model). 
Let $\lmaxi$ be the load that can be effectively handled on a server $i$ (beyond $\lmaxi$ the server fails due to, e.g., trashing). 
Let $\lmax = \max_{i} \lmaxi$. 
Let $\ltot = \sum_i n_i$ be the total load in the system. 
We assume that the total load can be effectively handled, $\sum_i \lmaxi \geq \ltot$ (otherwise, the system is clearly overloaded). We assume that the values $\lmaxi$ are chosen so that $f_i(\lmaxi)$ are equal to each other (equal to the maximal allowed processing time of the request). 

We assume that the load function $f_{i}$ is bounded on the interval $[0, \lmaxi]$ (If $l > \lmaxi$ then we follow the convention that $f_i(l) = \infty$). We assume $f_i$ is non-decreasing as when the load increases, requests are not processed faster.
We also assume that $f_i$ is convex and twice-differentiable on the interval $[0; \lmaxi]$ (functions that are not twice-differentiable can be well approximated by twice-differentiable functions). 
We assume that the first derivatives $f'_i$ of all $f_i$ are upper bounded by $U_1$ ($U_1 = \max_{i, l} f'_i(l)$), 
and that the second derivatives $f''_i$ are upper bounded by $U_2$ ($U_2 = \max_{i, l} f''_i(l)$).
These assumptions are technical---every function that is defined on the closed interval can be upper-bounded by a constant (however the complexity of our algorithms depends on these constants).

%These assumptions correspond to upper bounds on the rate of increase of the processing time; as the load function is defined only for a bounded load interval $[0, \lmaxi]$ and load $\lmaxi$ can be effectively handled, it does not make sense to consider unbounded derivatives.

% We also assume that $f_{i}(l)$ is lower bounded by $L > 0$ and that $f'_i(0) > \epsilon$ (If this is not the case we can consider a function $g_i(\ell) = f_{i}(\ell) + \epsilon \ell + \epsilon$ that arbitrarily well approximates $f_i$). 

%Most of the assumptions we use are technical. The class of functions that are two-times differentiable is wide; also, other functions can be well approximated by the differentiable functions.  Every function that is defined on the closed interval can be upper-bounded by a constant (however the complexity of our algorithms depends on that constant). 
%The class of functions that we consider generalizes the functions already considered in the literature, which we describe in Section~\ref{sec:specialCases}.

\medskip

\noindent
\textbf{Communication delays.}\quad
If the request is sent over the network, the observed handling time is increased by the communication latency on the link.
We denote the communication latency between $i$-th and $j$-th server as $c_{ij}$ (with $c_{ii}=0$). 
We assume that the requests are small, and so the communication delay of a single request does not depend on the amount of exchanged load (the same assumption was made in the previous works~\cite{Penmatsa:2006:CLB:1898953.1899089, Liu:2011:GGL:1993744.1993767, Skowron:2013:NDL:2510648.2510769, Grosu:2008:CLB:1455689.1455695, gallet09divisibleload, beaumont2005scheduling, veeravalli2002efficient, drozdowski2008scheduling} and it is confirmed by the experiments conducted on PlanetLab~\cite{Skowron:2013:NDL:2510648.2510769}). Thus, $c_{ij}$ is just a constant instead of a function of the network load.

We assume \emph{efficient $\epsilon$-load processing}: for sufficiently small load $\epsilon \to 0$ the processing time is lower than the communication latency, 
so it is not profitable to send the requests over the network. 
Thus, for any two servers $i$ and $j$ we have:
\begin{equation}
h_i(\epsilon) < \epsilon c_{ij} + h_j(\epsilon) \text{.}
\end{equation}
We use an equivalent formulation of the above assumption (as $h_i(0) = h_j(0) = 0$):
\begin{equation}
\frac{h_i(\epsilon) - h_i(0)}{\epsilon} < c_{ij} + \frac{h_j(\epsilon) - h_j(0)}{\epsilon} \text{.}
\end{equation}
Since the above must hold for every sufficiently small $\epsilon \to 0$, we get:
\begin{equation}
h_i'(0) < c_{ij} + h_j'(0) \Leftrightarrow f_i(0) < c_{ij} + f_j(0) \text{.}
\end{equation}

\medskip

\noindent
\textbf{Problem formulation: the total processing time.}\quad
We consider a system in which all requests have the same importance. Thus, 
the optimization goal is to minimize the total processing time of all requests $\sum C_i$, considering both the communication latencies and the requests' handling times on all servers, i.e., 
\begin{equation}
\sum{C_{i}} = \sum_{i = 1}^{m}l_{i}f_{i}(l_i) + \sum_{i = 1}^{m}\sum_{j = 1}^{m}c_{ij}r_{ij}\textrm{.}
\end{equation}
\medskip

We formalize our problem in the following definition:
\begin{definition}[Load balancing]
Given $m$ servers with initial loads $\{n_{i}, 0 \leq i \leq m\}$, load functions $\{f_i\}$ and  communication delays $\{c_{ij}: 0 \leq i,j \leq m\}$ find $\rho$, a vector of fractions, that minimizes the total processing time of the requests, $\sum{C_{i}}$.
\end{definition}

We denote the optimal relay fractions by $\rho^{*}$ and the load of the server $i$ in the optimal solution as $l^{*}_i$.

\subsection{Motivation}\label{sec:specialCases}

Since the assumptions about the load functions are moderate, our analysis is applicable to many systems. In order to apply our solutions one only needs to find the load functions $f_i$. In particular, our model generalizes the following previous models.

\medskip
\noindent
\textbf{Queuing model.}\quad
Our results generalize the results of Liu~et~al.~\cite{Liu:2011:GGL:1993744.1993767} for the queuing model.
In the queuing model, the initial load $n_i$ corresponds to the rate of local requests at the $i$-th server. Every server $i$ has a processing rate $\mu_i$. According to the queuing theory the dependency between the load $l$ (which is the effective rate of incoming requests) and the service time of the requests is described by $f_i(l) = \frac{1}{\mu_i - l}$~\cite{Gross:2008:FQT:1972549}.
% This function is lower bounded by $L = \min_i \frac{1}{\mu_i}$. 
Its derivative, $f'_i(l) = \frac{1}{(\mu_i - l)^2}$ is upper bounded by $U_1 = \max_i \frac{1}{(\mu_i - \lmaxi)^2}$, and its second derivative $f''_i(l) = \frac{2}{(l - \mu_i)^3}$ is upper bounded by $U_2 = \max_i \frac{2}{(\lmaxi - \mu_i)^3}$.

\medskip
\noindent
\textbf{Batch model.}\quad
Skowron~and~Rzadca~\cite{Skowron:2013:NDL:2510648.2510769} consider a model inspired by batch processing in high-performance computing. 
The goal is to minimize the flow time of jobs in a single batch.
In this case the function $f_i$ linearly depends on load $f_i(l) = \frac{l}{2s_i}$ (where $s_i$ is the speed of the $i$-th server). 
% This function is lower bounded by $L = \epsilon$. 
Its derivative is constant, and thus upper bounded by $\frac{1}{2s_i}$. The second derivative is equal to 0.

\section{Characterization of the problem}
In this section, we show various results that characterize the solutions in both the single-hop and the multiple-hop models. We will use these results in performance proofs in the next sections.

The relation between the single-hop model and the multiple-hop model is given by the proposition followed by the following lemma.

\begin{lemma}\label{lemma:recevAndSend}
If communication delays satisfy the triangle inequality (i.e., for every $i, j$, and $k$ we have $c_{ij} < c_{ik} + c_{kj}$), then in the optimal solution there is no server $i$ that both sends and a receives the load, i.e. there is no server $i$ such that $\exists_{j \neq i,k \neq i}\left( (\rho_{ij} > 0) \wedge (\rho_{ki} > 0) \right)$
\end{lemma}
\begin{proof}
For the sake of contradiction let us assume that there exist servers $i, j$ and $k$, such that $\rho_{ij} > 0$ and $\rho_{ki} > 0$. Then, if we modify the relay fractions: $\rho_{ij} := \rho_{ij} - \min(\rho_{ij}, \rho_{ki})$, $\rho_{jk} := \rho_{jk} - \min(\rho_{ij}, \rho_{ki})$, and $\rho_{kj} := \rho_{kj} + \min(\rho_{ij}, \rho_{ki})$, then the loads $l_i, l_j$ and $l_k$ remain unchanged, but the communication delay is changed by:
\begin{align*}
\min(\rho_{ij}, \rho_{ki})(c_{kj} - c_{ki} - c_{ij}) \textrm{,}
\end{align*}
which is, by the triangle inequality, a negative value. This completes the proof.
\end{proof}

\begin{proposition}\label{prop:multiSingleHopEquiv}
If communication delays satisfy the triangle inequality then the single-hop model and the multiple-hop model are equivalent.
\end{proposition}
\begin{proof}
From Lemma~\ref{lemma:recevAndSend} we get that in the optimal solution if $\rho_{ij} > 0$, then for every $k$ we have $r_{ki} = 0$. Thus, $n_i + \sum_{k}r_{ki} = n_i$.
\end{proof}

We will also use the following simple observation.

\begin{corollary}\label{prop:multiBellerSingleHop}
The total processing time in the multiple-hop model is not higher than in the single-hop model.
\end{corollary}

In the next two statements we recall two results given by Liu~et~al.~\cite{Liu:2011:GGL:1993744.1993767} (these results were formulated for the general load functions). 
First, there exists an optimal solution in which only $(2m-1)$ relay fractions $\rho_{ij}$ are positive. This theorem makes our analysis more practical: the optimal load balancing can be achieved with sparse routing tables. However, we note that most of our results are also applicable to the case when every server is allowed to relay its requests only to a (small) subset of the servers; in such case we need to set the communication delays between the disallowed pairs of servers to infinity. 
% and thus the servers need to maintain only limit information to implement optimal load balancing. 

\begin{theorem}[Liu~et~al.~\cite{Liu:2011:GGL:1993744.1993767}]
In a single-hop model there exists an optimal solution in which at most $(2m-1)$ relay fractions $\rho_{ij}$ have no-zero values.
\end{theorem}

Second, all optimal solutions are equivalent:
\begin{theorem}[Liu~et~al.~\cite{Liu:2011:GGL:1993744.1993767}]
Every server $i$ has in all optimal solutions the same load $l^{*}_i$.
\end{theorem}

Finally, in the next series of lemmas we characterize the optimal solution, by linear equations. We will use these characterization in the analysis of the central algorithm.

\begin{lemma}\label{lemma:linProg1}
In the multiple hop model, the optimal solution $\langle \rho_{ij}^{*} \rangle$ satisfies the following constraints:
\begin{alignat}{2}
 & \forall_{i}  \ & \;\; \ &    l_{i}^{*} \leq  \lmaxi  \label{in:linProgram0} \\ 
 & \forall_{i, j} \ & \; \; \ &  \rho^{*}_{ij}  \geq  0 \label{in:linProgram1} \\ 
 & \forall_{i} \ & \; \; \ &     \sum_{j=1}^{m} \rho^{*}_{ij} = 1 \text{.} \label{in:linProgram2} 
\end{alignat}
\end{lemma}
\begin{proof}
Inequality~\ref{in:linProgram0} ensures that the completion time of the requests is finite.
Inequalities~\ref{in:linProgram1}~and~\ref{in:linProgram2} state that the values of $\rho^{*}_{ij}$ are valid relay fractions.
\end{proof}

\begin{lemma}\label{lemma:linProg2}
In the multiple hop model, the optimal solution $\langle \rho_{ij}^{*} \rangle$ satisfies the following constraint:
\begin{align}
\forall_{i, j} \; \; f_{j}(l^{*}_j) + l^{*}_jf'_{j}(l^{*}_j) + c_{ij} \geq f_{i}(l^{*}_i) + l^{*}_if'_{i}(l^{*}_i) \label{in:linProgram3}
\end{align}
\end{lemma}
\begin{proof}
For the sake of contradiction let us assume that $f_{j}(l^{*}_j) + l^{*}_jf'_{j}(l^{*}_j) + c_{ij} < f_{i}(l^{*}_i) + l^{*}_if'_{i}(l^{*}_i)$. Since we assumed that $f_{j}(0) + c_{ij} > f_{i}(0)$ (see Section~\ref{sec:mathematical-model}), we infer that $l^{*}_i > 0$.

Next, we show that if $f_{j}(l^{*}_j) + l^{*}_jf'_{j}(l^{*}_j) + c_{ij} <  f_{i}(l^{*}_i) + l^{*}_if'_{i}(l^{*}_i)$ and $l^{*}_i > 0$, then the organization $i$ can improve the total processing time of the requests $\sum{C_i}$ by relaying some more load to the $j$-th server (which will lead to a contradiction).
Let us consider a function $F(\Delta r)$ that quantifies $i$'s and $j$'s contribution to $\sum{C_i}$ if $\Delta r$ requests are additionally send from $i$ to $j$ (and also takes into account the additional communication latency $\Delta r c_{ij}$):
\begin{align*}
F(\Delta r) = (l^{*}_i - \Delta r)f_{i}(l^{*}_i - \Delta r) + (l^{*}_j + \Delta r)f_{j}(l^{*}_j + \Delta r) + \Delta r c_{ij} \text{.}
\end{align*}
If $F(\Delta r) < F(0)$, then transferring extra $\Delta r$ requests from $i$ to $j$ decreases $\sum{C_i}$ (thus leading to a better solution). We compute the derivative of $F$:
\begin{align*}
F'(\Delta r) = -f_{i}(l^{*}_i - \Delta r) - (l^{*}_i - \Delta r)f'_{i}(l^{*}_i - \Delta r)  + f_{j}(l^{*}_j + \Delta r) + (l^{*}_j + \Delta r)f'_{j}(l^{*}_j + \Delta r) + c_{ij} \text{.}
\end{align*}
Since we assumed that $f_{j}(l^{*}_j) + l^{*}_jf'_{j}(l^{*}_j) + c_{ij} < f_{i}(l^{*}_i) + l^{*}_if'_{i}(l^{*}_i)$, we get that:
\begin{align*}
F'(0) = -f_{i}(l^{*}_i) - l^{*}_if'_{i}(l^{*}_i)  + f_{j}(l^{*}_j) + l^{*}_jf'_{j}(l^{*}_j) + c_{ij} < 0 \text{.}
\end{align*}
Since $F'$ is differentiable, it is continuous; so there exists $\Delta r_0 > 0$ such that $F'$ is negative on $[0; \Delta r_0]$, and thus $F$ is decreasing on $[0; \Delta r_0]$. Consequently, $F(\Delta r_0) < F(0)$, which contradicts the optimality of $\langle \rho_{ij}^{*} \rangle$.
\end{proof}

\begin{lemma}\label{lemma:linProg3}
In the multiple hop model, the optimal solution $\langle \rho_{ij}^{*} \rangle$ satisfies the following constraint:
\begin{align}
\forall_{i, j} \; \; \; \text{if} \; \rho^{*}_{ij} > 0 \; \text{then} \; f_{j}(l^{*}_j) + l^{*}_jf'_{j}(l^{*}_j) + c_{ij} \leq f_{i}(l^{*}_i) + l^{*}_if'_{i}(l^{*}_i) \label{in:linProgram4}
\end{align}
\end{lemma}
\begin{proof}
If $\rho^{*}_{ij} > 0$, then in the optimal solution $i$ sends some requests to $j$. There are two possibilities. Either some of the transferred requests of $i$ are processed on $j$, or $j$ sends all of them further to another server $j_2$. Similarly, $j_2$ may process some of these requests or send them all further to $j_3$. Let $j, j_2, j_3, \dots, j_{\ell}$ be the sequence of servers such that every server from $j, j_2, j_3, \dots, j_{\ell-1}$ transfers all received requests of $i$ to the next server in the sequence and $j_{\ell}$ processes some of them on its own.

First, we note that every server from $j, j_2, j_3, \dots, j_{\ell-1}$ has non-zero load. Indeed if this is not the case then let $j_{0}$ be the last server from the sequence which has load equal to 0. However we assumed that for sufficiently small load, it is faster to process it locally than to send it over the network to the next server $j_{k}$ ($f_{j_{0}}(\epsilon) < f_{k}(\epsilon) + c_{j_{0}k}$). This contradicts the optimality of the solution and shows that our observation is true.

Then, we take some requests processed on $j_{\ell-1}$ and swap them with the same number of requests owned by $i$, processed on $j_{\ell}$. After this swap $j_{\ell-1}$ processes some requests of $i$; such a swap does not change $\sum{C_i}$. Next, we repeat the same procedure for $j_{\ell-1}$ and $j_{\ell-2}$; then $j_{\ell-2}$ and $j_{\ell - 3}$; and so on. As a result, $j$ processes some requests of $i$.

The next part of the proof is similar to the proof of Lemma~\ref{lemma:linProg2}. Let us consider the function $G(\Delta r)$ that quantifies $i$'s and $j$'s contribution to $\sum{C_i}$ if $\Delta r$ requests are \emph{moved back} from $j$ to $i$ (i.e., not sent from $i$ to $j$):
\begin{align*}
G(\Delta r) = (l^{*}_i + \Delta r)f_{i}(l^{*}_i + \Delta r) + (l^{*}_j - \Delta r)f_{j}(l^{*}_j - \Delta r) - \Delta r c_{ij} \text{.}
\end{align*}
If $G(\Delta r) < G(0)$, executing $\Delta r$ requests on $i$ (and not on $j$) reduces $\sum{C_i}$.

$G(\Delta r) = F(-\Delta r)$ (see the proof of Lemma~\ref{lemma:linProg2}).
Thus, $G'(\Delta r) = -F'(\Delta r)$, and
\begin{equation}
G'(0) = -F'(0) = 
f_{i}(l^{*}_i) + l^{*}_if'_{i}(l^{*}_i)  - f_{j}(l^{*}_j) - l^{*}_jf'_{j}(l^{*}_j) - c_{ij} \text{.}
\end{equation}
As $l^*_i$ is optimal, $G'(0) \geq 0$, thus $f_{i}(l^{*}_i) + l^{*}_if'_{i}(l^{*}_i)  - f_{j}(l^{*}_j) - l^{*}_jf'_{j}(l^{*}_j) - c_{ij} \geq 0$, which proves the thesis.
\end{proof}

\begin{lemma}\label{lemma:linProg}
If some solution $\langle \rho_{ij} \rangle$ satisfies Inequalities~\ref{in:linProgram0},~\ref{in:linProgram1},~\ref{in:linProgram2},~\ref{in:linProgram3},~and~\ref{in:linProgram4} then every server $i$ under $\langle \rho_{ij} \rangle$ has the same load as in the optimal solution $\langle \rho^{*}_{ij} \rangle$.
\end{lemma}
\begin{proof}
Let $S_{+}$ denote the set of servers that in $\rho^{*}$ have greater or equal load than in $\rho$ ($l_i^{*} \geq l_i$). For the sake of contradiction let us assume that $S_{+}$ is non-empty and that it contains at least one server $i$ that in $\rho^{*}$ has strictly greater load than in $\rho$ ($l_i^{*} > l_i$).

Let $j \in S_{+}$; we will show that $j$ in $\rho^{*}$ can receive requests only from the servers from $S_{+}$.
By definition of $S_+$, $l_j^{*} \geq l_j$. Consider a server $i$ that in $\rho^{*}$ relays some of its requests to $j$; we will show that $l_i^{*} \geq l_i$. Indeed, since $\rho^{*}_{ij} > 0$, from Inequality~\ref{in:linProgram4} we get that:
\begin{equation}\label{eq:linProg-der1}
f_{j}(l^{*}_j) + l^{*}_jf'_{j}(l^{*}_j) + c_{ij} \leq f_{i}(l^{*}_i) + l^{*}_if'_{i}(l^{*}_i) \text{.}
\end{equation}
Since we assumed that $\langle \rho_{ij} \rangle$ satisfies Inequality~\ref{in:linProgram3}, we get
\begin{equation}\label{eq:linProg-der2}
f_{j}(l_j) + l_jf'_{j}(l_j)+ c_{ij} \geq f_{i}(l_i) + l_if'_{i}(l_i) \text{.}
\end{equation}
By combining these relations we get:
\begin{align*}
f_{i}(l^{*}_i) + l^{*}_if'_{i}(l^{*}_i) & \geq f_{j}(l^{*}_j) + l^{*}_jf'_{j}(l^{*}_j) + c_{ij} & \text{from Eq.~\ref{eq:linProg-der1}}\\
                                        & \geq f_{j}(l_j) + l_jf'_{j}(l_j) + c_{ij} & \text{as $l^*_j \geq l_j$ and $f_j$ is convex and non-decreasing }\\
                                        & \geq f_{i}(l_i) + l_if'_{i}(l_i) & \textrm{from Eq.~\ref{eq:linProg-der2}.}
\end{align*}
Since $f_i$ is convex, the function $f_i(l) + lf_i'(l)$ is non-decreasing (as the sum of two non-decreasing functions); thus $l_i^{*} \geq l_i$.

Similarly, we show that any $i \in S_{+}$ in $\rho$ can send requests only to other $S_+$ servers. 
Consider a server $j$ that in $\rho$ receives requests from $i$. 
% Since $\rho_{ij} > 0$, from Inequality~\ref{in:linProgram4} we know that
% $f_{j}(l_j) + l_jf'_{j}(l_j) + c_{ij} \leq f_{i}(l_i) + l_if'_{i}(l_i)$.
% From the optimality of $\rho^{*}$, from Inequality~\ref{in:linProgram3} we get that
% $f_{j}(l^{*}_j) + l^{*}_jf'_{j}(l^{*}_j) + c_{ij} \geq f_{i}(l^{*}_i) + l^{*}_if'_{i}(l^{*}_i)$. Thus, by combining the two inequalities we get:
\begin{align*}
f_{j}(l^{*}_j) + l^{*}_jf'_{j}(l^{*}_j) & \geq f_{i}(l^{*}_i) + l^{*}_if'_{i}(l^{*}_i) - c_{ij} & \text{Eq.~\ref{in:linProgram3}}\\
                                        & \geq f_{i}(l_i) + l_if'_{i}(l_i) - c_{ij} & \text{as $l^*_i \geq l_i$ and $f_i$ is convex and non-decreasing} \\
                                        & \geq f_{j}(l_j) + l_jf'_{j}(l_j) & \textrm{as $\rho_{ij}>0$, from Eq.~\ref{in:linProgram4}.}
\end{align*}
Thus, $l_j^{*} \geq l_j$.

Let $l_{in}$ be the total load sent in $\rho$ to the servers from $S_{+}$ by the servers outside of $S_{+}$. Let $l_{out}$ be the total load sent by the servers from $S_{+}$ in $\rho$ to the servers outside of $S_{+}$. Analogously we define $l^{*}_{in}$ and $l^{*}_{out}$ for the state $\rho^{*}$. In the two previous paragraphs we showed that $l^{*}_{in} = 0$ and that $l_{out} = 0$. However, since the total load of the servers from $S_{+}$ is in $\rho^{*}$ greater than in $\rho$, we get that:
\begin{align*}
l^{*}_{in} - l^{*}_{out} > l_{in} - l_{out} \textrm{.}
\end{align*}
From which we get that: $-l^{*}_{out} > l_{in}$, i.e. $l_{in} + l^*_{out} < 0$, which leads to a contradiction as transfers $l_{in}$ and $l^*_{out}$ cannot be negative.
\end{proof}

\section{An approximate centralized algorithm}\label{sec:approximate-centralized}

In this section we show a centralized algorithm for the multiple-hop model. As a consequence of Proposition~\ref{prop:multiSingleHopEquiv} the results presented in this section also apply to the single-hop model with the communication delays satisfying the triangle inequality.

For the further analysis we introduce the notion of optimal network flow.

\begin{definition}[Optimal network flow]
The vector of relay fractions $\rho = \langle \rho_{ij} \rangle$ has an optimal network flow if and only if there is no $\rho' = \langle \rho'_{ij} \rangle$ such that every server in $\rho'$ has the same load as in $\rho$ and such that the total communication delay of the requests $\sum_{i, j}c_{ij}r'_{ij}$ in $\rho'$ is lower than the total communication delay $\sum_{i, j}c_{ij}r_{ij}$ in $\rho$.
\end{definition} 

The problem of finding the optimal network flow reduces to finding a minimum cost flow in uncapacitated network.
Indeed, in the problem of finding a minimum cost flow in uncapacitated network we are given a graph with the cost of the arcs and demands (supplies) of the vertices. For each vertex $i$, $b_i$ denotes the demand (if positive) or supply (if negative) of $i$. We look for the flow that satisfies demands and supplies and minimizes the total cost. To transform our problem of finding the optimal network flow to the above form it suffices to set $b_{i} = l_i - n_i$. 
Thus our problem can be solved in time $O(m^3\log m)$~\cite{Orlin88afaster}. Other distributed algorithms include the one of Goldberg~et~al.~\cite{minimumCirculation}, and the asynchronous auction-based algorithms~\cite{auctionBasedMinCostFlow}, with e.g., the complexity of $O(m^3\log (m)\log (\max_{i,j}c_{ij}))$.

The following theorem estimates how far is the current solution to the optimal based on the degree to which Inequality~\ref{in:linProgram3} is not satisfied. We use the theorem to prove approximation ratio of the load balancing algorithm.

\begin{theorem}\label{thm:approxBound}
Let $\rho$ be the vector of relay fractions satisfying Inequalities~\ref{in:linProgram0},~\ref{in:linProgram1},~\ref{in:linProgram2}~and~\ref{in:linProgram4}, and having an optimal network flow. Let $\Delta_{ij}$ quantify the extent to which Inequality~\ref{in:linProgram3} is not satisfied:
\begin{align*}
\Delta_{ij} = \max(0, f_{i}(l_i) + l_if'_{i}(l_i) - f_{j}(l_j) - l_jf'_{j}(l_j) - c_{ij}) \textrm{.}
\end{align*}
Let $\Delta = \max_{i, j}\Delta_{ij}$.
Let $e$ be the absolute error---the difference between $\sum{C_i}$ for solution $\rho$ and for  $\rho^{*}$, $e = \sum{C_i}(\rho) - \sum{C_i}(\rho^*)$.
For the multiple-hop model and for the single-hop model satisfying the triangle inequality we get the following estimation:
\begin{align*}
e \leq \ltot m \Delta\textrm{.}
\end{align*}
\end{theorem}

\begin{proof}
Let $I$ be the problem instance. Let $\tilde{I}$ be a following instance: initial loads $n_i$ in $\tilde{I}$ are the same as in $I$; communication delays $c_{ij}$ are increased by $\Delta_{ij}$ ($\tilde{c_{ij}} := c_{ij} + \Delta_{ij}$). Let $\tilde{\rho^*}$ be the optimal solutions for $\tilde{I}$ in the multiple-hop model.

By Lemma~\ref{lemma:linProg}, loads of servers in $\rho$ are the same as in $\tilde{\rho^*}$, as $\rho$ satisfies all inequalities for $\tilde{I}$. 
% First, we show that $\tilde{\rho^*}$ in $\tilde{I}$ may not have smaller total communication delay than $\rho$ in $I$. 
Let $c^{*}$ and $c$ denote the total communication delay of $\tilde{\rho^*}$ in $\tilde{I}$ and $\rho$ in $I$, respectively.
First, we show that $c^{*} \geq c$.

For the sake of contradiction, assume that $c^{*} < c$. We take the solution $\tilde{\rho^*}$ in $\tilde{I}$ and modify $\tilde{I}$ by decreasing each latency $\tilde{c_{ij}}$ by $\Delta_{ij}$. We obtain instance $I$.
During the process, we decreased (or did not change) communication delay over every link, and so we decreased (or did not change) the total communication delay. Thus, in $I$, $\tilde{\rho^*}$ has smaller communication delay than $\rho$. This contradicts the thesis assumption that $\rho$ had in $I$ the optimal network flow. 

As $\tilde{I}$ has the same initial loads and not greater communication delay,
\begin{align*}
\sum{C_i}(\rho, I) \leq \sum{C_i}(\tilde{\rho^*}, \tilde{I}) \textrm{.}
\end{align*}
Based on Proposition~\ref{prop:multiSingleHopEquiv}, the same result holds if $\rho$ is the solution in the single-hop model satisfying the triangle inequality.

%Let us once again take a look at the process of moving from $\tilde{I}$ to $I$. In such a process we decreased the communication delay of every link by at most $\Delta$. Also, 

We use a similar analysis to bound the processing time.
In the multiple-hop model, if the network flow is optimal, then every request can be relayed at most $m$ times. 
Thus, any solution transfers at most $\ltot m$ load. 
Thus, by increasing latencies from $I$ to $\tilde{I}$ we increase the total communication delay of a solution by at most $\ltot m \Delta$. Taking the optimal solution $\rho^*$, we get:
\begin{align*}
\sum{C_i}(\rho^{*}, \tilde{I}) \leq l_{\mathit{tot}} m \Delta + \sum{C_i}(\rho^{*}, I) \textrm{.}
\end{align*}

As $\sum{C_i}(\tilde{\rho^*}, \tilde{I}) \leq \sum{C_i}(\rho^{*}, \tilde{I})$, by combining the two inequalities we get:
\begin{align*}
\sum{C_i}(\rho, I) \leq \sum{C_i}(\tilde{\rho^*}, \tilde{I}) \leq \sum{C_i}(\rho^{*}, \tilde{I}) \leq l_{\mathit{tot}} m \Delta + \sum{C_i}(\rho^{*}, I) \textrm{.}
\end{align*}
\end{proof}

The above estimations allow us to construct an approximation algorithm (Algorithm~\ref{alg:multipleHopApprox}). The lines~\ref{algline::init1}~to~\ref{algline::init2} initialize the variables. In line~\ref{algline::finiteSol} we build any finite solution (any solution for which the load $l_i$ on the $i$-th server does not exceed $\lmaxi$). Next in the while loop in line~\ref{algline::whileLoop} we iteratively improve the solution. In each iteration we find the pair $(i, j)$ with the maximal value of $\Delta_{ij}$. Next we balance the servers $i$ and $j$ in line~\ref{algline::adjust1}. Afterwards, it might be possible that the current solution does not satisfy Inequality~\ref{in:linProgram4}. In the lines~\ref{algline::topSort}~to~\ref{algline::adjust2} we fix the solution so that Inequality~\ref{in:linProgram4} holds.

\begin{algorithm}[th!]
   \footnotesize
   \SetKwInput{KwNotation}{Notation}
   \SetKwFunction{Improve}{Improve}
   \SetKwFunction{Main}{Main}
   \SetKwFunction{Adjust}{Adjust}
   \SetKwFunction{AdjustBack}{AdjustBack}
   \SetKwFunction{PriorityQueue}{PriorityQueue}
   \SetKwFunction{OptimizeNetwork}{OptimizeNetworkFlow}
   \SetKwFunction{BuildAnyFiniteSolution}{BuildAnyFiniteSolution}
   \SetKwFunction{OptimizeNetworkA}{\textbf{OptimizeNetworkFlow}}
   \SetKwBlock{Block}
   \SetAlCapFnt{\footnotesize}
   \KwNotation{\\
	$\pmb{e}$ --- the required absolute error of the algorithm.\\
	$\pmb{c_{ij}}$ --- the communication delay between $i$-th and $j$-th server.\\
	$\pmb{l}[i]$ --- the load of the $i$-th server in a current solution.\\
	$\pmb{r}$$[i, j]$ --- the number of requests relayed between $i$-th and $j$-th server in a current solution.\\
	\textbf{\OptimizeNetworkA}($\rho$, $\langle c_{ij} \rangle$) --- builds an optimal network flow using algorithm of Orlin~\cite{Orlin88afaster}. 
     }
	\hspace{5mm}

	\Adjust{$i, j$}:
	\Block{
		$\Delta r \leftarrow \mathrm{argmin}_{\Delta r} \left((l_i - \Delta r)f_{i}(l_i - \Delta r) + (l_j + \Delta r)f_{j}(l_j + \Delta r) + \Delta r c_{ij} \right)$\;
		$l[i] \leftarrow l[i] - \Delta r$\;
		$l[j] \leftarrow l[j] + \Delta r$\;
		$r[i, j] \leftarrow r[i, j] + \Delta r$\;
	}

	\Improve{$i, j$}:
	\Block{
		\Adjust($i, j$)\nllabel{algline::adjust1}\;
		$\mathit{servers} \leftarrow$ sort servers topologically according to the order $\prec$: $i \prec j \iff \rho_{ij} > 0$\nllabel{algline::topSort}\;
		\SetKw{KwTo}{in}
		\For{$\ell$ \KwTo $\mathit{servers}$}
		{
			\SetKw{KwTo}{to}
			\For{$k \leftarrow 1$ \KwTo $m$}
			{
				\SetKw{KwTo}{and}
				\If{$r[k, \ell] > 0$ \KwTo $f_{\ell}(l^{*}_{\ell}) + l^{*}_{\ell}f'_{\ell}(l^{*}_{\ell}) + c_{k\ell} > f_{k}(l^{*}_k) + l^{*}_kf'_{k}(l^{*}_k)$}
				{
					\SetKw{KwTo}{to}
					\AdjustBack($\ell, k$)\nllabel{algline::adjust2}\;
				}
			}
		}
	}

	\Main{$\langle c_{ij} \rangle$, $\langle n_{i} \rangle$, $\langle s_{i} \rangle$}:
	\Block{
		\For{$i \leftarrow 1$ \KwTo $m$ \nllabel{algline::init1}}
		{
			$l[i] \leftarrow n_i$\;
			\For{$j \leftarrow 1$ \KwTo $m$}
			{
				$r[i, j] \leftarrow 0$\;
			}
			$r[i, i] \leftarrow n_i$\nllabel{algline::init2}\;
		}
		\BuildAnyFiniteSolution{}  \nllabel{algline::finiteSol}\;
		\OptimizeNetwork{$r$, $\langle c_{ij} \rangle$}\;
		$(i, j) \leftarrow \mathrm{argmax}_{(i, j)}\Delta_{ij}$\;
		% \tcc{in case of single-hop model we can use $\Delta_{ij} > \frac{\epsilon}{\ltot}$ in condition below}
		\While{$\Delta_{ij} > \frac{e}{\ltot m}$\nllabel{algline::whileLoop}} 
		{
			$(i, j) \leftarrow \mathrm{argmax}_{(i, j)}\Delta_{ij}$\;
			\Improve{$i$, $j$};
		}
		\OptimizeNetwork{$r$, $\langle c_{ij} \rangle$}\;
	}

   \caption{The approximation algorithm for multiple-hop model.}
   \label{alg:multipleHopApprox}
\end{algorithm}

The following Theorem shows that Algorithm~\ref{alg:multipleHopApprox} achieves an arbitrary small absolute error $e$.
%The relative error is the ratio of the absolute error $e$ of the algorithm's result $\rho$
%(the difference in $\sum{C_i}$ between the solution returned by the algorithm and the optimal solution) 
%to the optimal $\sum{C_i}^*$, $e_r = ( \sum{C_i}(\rho) - \sum{C_i}^* ) / \sum{C_i}^*$.

\begin{theorem}\label{thm:approxQuality}
Let $e_d$ be the desired absolute error for Algorithm~\ref{alg:multipleHopApprox}, and let $e_i$ be the initial error. In the multiple-hop model Algorithm~\ref{alg:multipleHopApprox} decreases the absolute error from $e_i$ to $e_d$ in time $O(\frac{\ltot^2m^4e_i(U_1 + \lmax U_2)}{e_d^2})$.
\end{theorem}
\begin{proof}
Let $l_i$ and $l_j$ be the loads of the servers $i$ and $j$ before the invocation of the \texttt{Adjust} function in line~\ref{algline::adjust1} of Algorithm~\ref{alg:multipleHopApprox}.
Let $\Delta_{ij}$ quantify how much Inequality~\ref{in:linProgram3} is not satisfied, $\Delta_{ij} =  f_{i}(l_i) + l_if'_{i}(l_i) - f_{j}(l_j) - l_jf'_{j}(l_j) - c_{ij}$.
As in proof of Lemma~\ref{lemma:linProg2}, consider a function $F(\Delta r)$ that quantifies $i$'s and $j$'s contribution to $\sum{C_i}$ if $\Delta r$ requests of are additionally send from $i$ to $j$:
\begin{align*}
F(\Delta r) = (l_i - \Delta r)f_{i}(l_i - \Delta r) + (l_j + \Delta r)f_{j}(l_j + \Delta r) + \Delta r c_{ij} \text{.}
\end{align*}
As previously, the derivative of $F$ is:
\begin{align*}
F'(\Delta r) = -f_{i}(l_i - \Delta r) - (l_i - \Delta r)f'_{i}(l_i - \Delta r)  + f_{j}(l_j + \Delta r) + (l_j + \Delta r)f'_{j}(l_j + \Delta r) + c_{ij} \text{.}
\end{align*}
Thus, $F'(0) = -\Delta_{ij}$.
The second derivative of $F$ is equal to:
\begin{align*}
F''(\Delta r) = 2f'_{i}(l_i - \Delta r) + (l_i - \Delta r)f''_{i}(l_i - \Delta r) + 2f'_{j}(l_j + \Delta r) + (l_j + \Delta r)f''_{j}(l_j + \Delta r) \text{.}
\end{align*}
The second derivative is bounded by:
\begin{align}
|F''(\Delta r)| \leq 4U_1 + 2\lmax U_2 \text{.} \label{eq:der2-upperbounded}
\end{align}
For any function $f$ with a derivative $f'$ bounded on range $[x_0, x]$ by a constant $f'_{\max}$, the value $f(x)$ is upper-bounded by:
\begin{equation}
f(x) \leq f(x_0) + (x - x_0) f'_{\max} \text{.} \label{eq:der-upperbound}
\end{equation}
Using this fact, we upper-bound the first derivative by:
\begin{align*}
F'(\Delta r) \leq F'(0) + \Delta r(4U_1 + 2\lmax U_2).
\end{align*}
We use a particular value of the load difference $\Delta r_0 = \frac{\Delta_{ij}}{8U_1 + 4\lmax U_2}$, getting that for $\Delta r \leq \Delta r_0$, we have: 
% As the second derivative is positive, the first derivative is increasing; thus, for $\Delta r \in [0, \Delta r_0]$, the first derivative can be bounded by:
\begin{align*}
F'(\Delta r) & \leq F'(0) + \Delta r(4U_1 + 2\lmax U_2) \\
  & \leq F'(0) + \Delta r_0(4U_1 + 2\lmax U_2) \\ 
  & \leq -\Delta_{ij} + \frac{\Delta_{ij}}{8U_1 + 4\lmax U_2} \cdot (4U_1 + 2\lmax U_2) \leq -\frac{1}{2}\Delta_{ij} \textrm{.}
\end{align*}

We can use Inequality~\ref{eq:der-upperbound} for a function $F$ to lower-bound the reduction in $\sum{C_i}$ for $\Delta r_0$ as $F(0) - F(\Delta r_0)$:
\begin{align*}
F(0) - F(\Delta r_0) \geq \frac{1}{2}\Delta_{ij}|r_0 - 0| = \frac{\Delta_{ij}}{8U_1 + 4\lmax U_2} \cdot \frac{1}{2}\Delta_{ij} = \frac{\Delta_{ij}^2}{16U_1 + 8\lmax U_2} \textrm{.}
\end{align*}

% From the above inequality, and from the fact that $F'(0) = -\Delta_{ij}$ we infer that on interval:
% \begin{align*}
% \Delta r \in [O, x] = [0, \frac{\Delta_{ij}}{8U_1 + 4\lmax U_2}]
% \end{align*}
% the derivative $F'$ can be bounded by:
% \begin{align*}
% F'(\Delta r) \leq F'(0) + (4U_1 + 2\lmax U_2)x = -\Delta_{ij} + \frac{\Delta_{ij}}{8U_1 + 4\lmax U_2} \cdot (4U_1 + 2\lmax U_2) \leq -\frac{1}{2}\Delta_{ij} \textrm{.}
% \end{align*}
% Thus, we get:
% \begin{align*}
% F(0) - F\left(\frac{\Delta_{ij}}{8U_1 + 4\lmax U_2}\right) \geq \frac{\Delta_{ij}}{8U_1 + 4\lmax U_2} \cdot \frac{1}{2}\Delta_{ij} = \frac{\Delta_{ij}^2}{16U_1 + 8\lmax U_2} \textrm{.}
% \end{align*}

To conclude that \texttt{Adjust} function invoked in line~\ref{algline::adjust1} reduces the total processing time by at least $\frac{\Delta_{ij}^2}{16U_1 + 8\lmax U_2}$, we still need ensure that the server $i$ has enough (at least $\Delta r_0 = \frac{\Delta_{ij}}{8U_1 + 4\lmax U_2}$) load to be transferred to $j$. However we recall that the value of $F'$ in $\Delta r_0$ is negative, $F'(\Delta r_0) < -\frac{1}{2}\Delta_{ij} < 0$. This means that after transferring $\Delta r_0$ requests, sending more requests from $i$ to $j$ further reduces $\sum{C_i}$.
Thus, if $i$'s load would be lower than $\Delta r_0$, this would contradict the efficient $\epsilon$-load processing assumption.

% If $i$ had only $\epsilon \to 0$ load then it would no longer be profitable (see the assumptions about load functions in Section~\ref{sec:mathematical-model}).

Also, every invocation of \texttt{AdjustBack} decreases the total completion time $\sum{C_i}$. Thus,  after invocation of \texttt{Improve} the total completion time $\sum{C_i}$ is decreased by at least $\frac{\Delta_{ij}^2}{16U_1 + 8\lmax U_2}$.

Each invocation of \texttt{Improve} preserves the following invariant: in the current solution Inequalities~\ref{in:linProgram1},~\ref{in:linProgram2}~and~\ref{in:linProgram4} are satisfied. It is easy to see that Inequalities~\ref{in:linProgram1},~\ref{in:linProgram2} are satisfied. We will show that Inequality~\ref{in:linProgram4} holds too. Indeed, this is accomplished by a series of invocations of \texttt{AdjustBack} in line~\ref{algline::adjust2}. Indeed, from the proof of Lemma~\ref{lemma:linProg3}, after invocation of the \texttt{Adjust} function for the servers $i$ and $j$, these servers satisfy Inequality~\ref{in:linProgram4}. 

We also need to prove that the servers can be topologically sorted in line~\ref{algline::topSort}, that is that there is no such sequence of servers $i_1, \dots, i_k$ that $r_{i_{j}i_{j+1}} > 0$ and $r_{i_{k}i_{1}} > 0$. For the sake of contradiction let us assume that there exists such a sequence. Let us consider the first invocation of \texttt{Adjust} in line~\ref{algline::adjust1} that creates such a sequence. Without loss of generality let us assume that such \texttt{Adjust} was invoked for the servers $i_k$ and $i_1$. This means that before this invocation $\Delta_{i_ki_1} > 0$, and so
$f_{{i_k}}(l_{i_k}) + l_{i_k}f'_{{i_k}}(l_{i_k}) > f_{{i_1}}(l_{i_1}) + l_{i_1}f'_{{i_1}}(l_{i_1})$.
Since the invariant was satisfied before entering \texttt{Adjust} and since $r_{i_{j}i_{j+1}} > 0$, from Inequality~\ref{in:linProgram4} we infer that
$f_{{i_{j+1}}}(l_{i_{j+1}}) + l_{i_{j+1}}f'_{{i_{j+1}}}(l_{i_{j+1}}) + c_{i_{j},i_{j+1}} \leq f_{{i_{j}}}(l_{i_{j}}) + l_{i_{j}}f'_{{i_{j}}}(l_{i_{j}})$, and so that
$f_{{i_{j+1}}}(l_{i_{j+1}}) + l_{i_{j+1}}f'_{{i_{j+1}}}(l_{i_{j+1}}) \leq f_{{i_{j}}}(l_{i_{j}}) + l_{i_{j}}f'_{{i_{j}}}(l_{i_{j}})$. Thus, we get contradiction:
\begin{align*}
f_{{i_{1}}}(l_{i_{1}}) + l_{i_{1}}f'_{{i_{1}}}(l_{i_{1}}) \geq
f_{{i_{2}}}(l_{i_{2}}) + l_{i_{2}}f'_{{i_{2}}}(l_{i_{2}}) \geq
\dots \geq
f_{{i_{k}}}(l_{i_{k}}) + l_{i_{k}}f'_{{i_{k}}}(l_{i_{k}}) \geq
f_{{i_{1}}}(l_{i_{1}}) + l_{i_{1}}f'_{{i_{1}}}(l_{i_{1}}) \textrm{.}
\end{align*}
Which proves that the invariant is true.

If the algorithm finishes, then $\Delta < \frac{e_d}{\ltot m}$. After performing the last step of the algorithm the network flow is optimized and we can use Theorem~\ref{thm:approxBound} to infer that the error is at most $e_d$.

We estimate the number of iterations to decrease the absolute error from $e_i$ to $e_{d}$. To this end, we estimated the decrease of the error after a single iteration of the while loop in line~\ref{algline::whileLoop}.
The algorithm continues the last loop only when $\Delta \geq \frac{e_d}{\ltot m}$.
Thus, after a single iteration of the loop the error decreases by at least $\frac{\Delta_{ij}^2}{16U_1 + 8\lmax U_2} \geq \frac{e_d^2}{\ltot^2m^2(16U_1 + 8\lmax U_2)}$.
Thus, after $O(\frac{\ltot^2m^2e_i(U_1 + \lmax U_2)}{e_d^2})$ iterations the error decreases to 0. Since every iteration of the loop has complexity $O(m^2)$, we get the thesis.
\end{proof}

Using a bound from Theorem~\ref{thm:approxBound} corresponding to the single-hop model we get the following analogous results.

\begin{corollary}
If the communication delays satisfy the triangle inequality then Algorithm~\ref{alg:multipleHopApprox} for the single-hop model decreases the absolute error from $e_i$ to $e_d$ in time $O(\frac{\ltot^2m^4 e_i(U_1 + \lmax U_2)}{e_d^2})$.
\end{corollary}

For the relative (to the total load) errors $e_{i,r} = \frac{e_i}{\ltot}$, and $e_{d,r} = \frac{e_d}{\ltot}$, Algorithm~\ref{alg:multipleHopApprox} decreases $e_{i,r}$ to $e_{d,r}$ in time $O(\frac{\ltot(U_1 + \lmax U_2)e_{i, r}}{e_{d, r}^2}m^4)$. Thus, we get the shortest runtime if $\ltot$ is large and $e_{i, r}$ is small. If the initial error $e_{i, r}$ is large we can use a modified algorithm that performs \texttt{OptimizeNetworkFlow} in every iteration of the last ``while'' loop (line~\ref{algline::whileLoop}). Using a similar analysis as before we get the following bound.

\begin{theorem}
The modified Algorithm~\ref{alg:multipleHopApprox} that performs \texttt{OptimizeNetworkFlow} in every iteration of the last ``while'' loop (line~\ref{algline::whileLoop}) decreases the relative error $e_{i,r}$ by a multiplicative constant factor in time $O(\frac{\ltot m^5 \log m(U_1 + \lmax U_2)}{e_{i, r}})$.
\end{theorem}
\begin{proof}
The analysis is similar as in the proof of Theorem~\ref{thm:approxQuality}. Here however at the beginning of each loop the network flow is optimized. If the absolute error before the loop is equal to $e$, then from Theorem~\ref{thm:approxBound} we infer that $\Delta \geq \frac{e}{\ltot m}$. Thus, after a single iteration of the loop the error decreases by $\frac{\Delta^2}{16U_1 + 8\lmax U_2} \geq \frac{e^2}{\ltot^2m^2(16U_1 + 8\lmax U_2)}$, and so by the factor of:
\begin{align*}
\left(e - \frac{e^2}{\ltot^2m^2(16U_1 + 8\lmax U_2)}\right)/e = \left(1 - \frac{e}{\ltot^2m^2(16U_1 + 8\lmax U_2)}\right) \text{.}
\end{align*}
Taking the relative error $e_{i,r}$ as $\frac{e}{\ltot}$ we get that every iteration decreases the relative error by a constant factor $\left(1 - \frac{e_{i, r}}{\ltot m^2(16U_1 + 8\lmax U_2)}\right)$.
Thus, after $O(\frac{\ltot m^2(U_1 + \lmax U_2)}{e_{i, r}})$ iterations the error decreases by a constant factor. Since the complexity of every iteration of the loop is dominated by the algorithm optimizing the network flow (which has complexity $O(m^3\log m)$), we get the thesis.
\end{proof}

Algorithm~\ref{alg:multipleHopApprox} is any-time algorithm. We can stop it at any time and get a so-far optimized solution.

\section{Distributed algorithm}\label{sec:distr-algor}

\begin{algorithm}[t!]
  \SetKwInOut{Input}{input}
  \Input{$(i, j)$ -- the identifiers of the two servers}
  \KwData{$\forall_{k}$ $r_{ki}$ -- initialized to the number of requests owned by $k$ and relayed to $i$ ($\forall_{k}$ $r_{kj}$ is defined analogously)}
  \KwResult{The new values of $r_{ki}$ and $r_{kj}$}
  \ForEach{$k$}{
       $r_{ki} \leftarrow r_{ki} + r_{kj}$; $r_{kj} \leftarrow 0$\; 
  }
  $l_{i} \leftarrow \sum_{k} r_{ki}$ ; $l_{j} \leftarrow 0$ \;

  $servers$ $\leftarrow$ sort $[k]$ so that $c_{kj} - c_{ki} < c_{k'j} - c_{k'i}$ $\implies$ $k$ is before $k'$\;
  \ForEach{$k \in servers$\nllabel{algline::secLoopStart}}
  {
      $\Delta_{\mathit{opt}} r_{ikj} \leftarrow \mathrm{argmin}_{\Delta r}\left( h_i(l_i - \Delta r) +  h_j(l_{j} + \Delta r) - \Delta r c_{ki} + \Delta r c_{kj}\right) $ \nllabel{algline::argmin}\;
      $\Delta r_{ikj} \leftarrow \min\left(\Delta_{\mathit{opt}} r_{ikj}, r_{ki}\right)$ \;
      \If{$\Delta r_{ikj} > 0$}
      {
         $r_{ki} \leftarrow r_{ki} - \Delta r_{ikj}$; $r_{kj} \leftarrow r_{kj} + \Delta r_{ikj}$ \;
         $l_{i} \leftarrow l_{i} - \Delta r_{ikj}$; $l_{j} \leftarrow l_{j} + \Delta r_{ikj}$ \nllabel{algline::secLoopEnd}\;
      }
  }
  \Return{for each $k$: $r_{ki}$ and $r_{kj}$}
  \caption{\textsc{calcBestTransfer}(i, j)}
  \label{alg:exchangingLoads}
\end{algorithm}

\begin{algorithm}[h]
  \SetKwFunction{calcBestTransfer}{calcBestTransfer}
  \SetKwFunction{transfer}{relay}
  \SetKwFunction{random}{random}
  \SetKwInOut{Notation}{Notation}
  partner $\leftarrow$ \random{m}\;
  \transfer(id, partner, \calcBestTransfer{id, partner})\;
  \caption{Min-Error (MinE) algorithm performed by server id.}
  \label{alg:distributedOptimal}
\end{algorithm}

The centralized algorithm requires the information about the whole network. The size of the input data is $O(m^{2})$. A centralized algorithm has thus the following drawbacks:
\begin{inparaenum}[(i)] 
\item collecting information about the whole network is time-consuming; moreover, loads and latencies may frequently change;
\item the central algorithm is more vulnerable to failures. 
Motivated by these limitations we introduce a distributed algorithm for optimizing the query processing time.
\end{inparaenum}

% The distributed algorithm requires that each server has an up-to-date information about the communication delays from itself to the other servers (and not for all pairs of servers). It is easy to measure pair-wise latencies (Section~\ref{sec::introduction}). For each server, the size of the input data is $O(m)$. 

Each server, $i$, keeps for each server, $k$, information about the number of requests that were relayed to $i$ by $k$. The algorithm iteratively improves the solution -- the $i$-th server in each step communicates with a random partner server -- $j$ (Algorithm~\ref{alg:distributedOptimal}). The pair $(i,j)$ locally optimizes the current solution by adjusting, for each $k$, $r_{ki}$ and $r_{kj}$ (Algorithm~\ref{alg:exchangingLoads}). 
In the first loop of the Algorithm~\ref{alg:exchangingLoads}, one of the servers $i$, takes all the requests that were previously assigned to $i$ and to $j$. Next, all the servers $[k]$ are sorted according to the ascending order of $(c_{kj} - c_{ki})$. The lower the value of $(c_{kj} - c_{ki})$, the less communication delay we need to pay for running requests of $k$ on $j$ rather than on $i$. Then, for each $k$, the loads are balanced between servers $i$ and $j$.  Theorem~\ref{lemma::delegationOptimality} shows that Algorithm~\ref{alg:exchangingLoads} 
optimally balances the loads on the servers $i$ and $j$.

% The algorithm requires only two servers in each optimization step, thus it is robust in presence of temporary failures. 

The idea of the algorithm is similar to the diffusive load balancing~\cite{conf/ipps/AdolphsB12, Ackermann:2009:DAQ:1583991.1584046, Berenbrink:2011:DSL:2133036.2133152}; however there are substantial differences related to the fact that the machines are geographically distributed: \begin{inparaenum}[(i)] 
\item In each step no real requests are transferred between the servers; this process can be viewed as a simulation run to calculate the relay fractions $\rho_{ij}$. Once the fractions are calculated the requests are transferred and executed at the appropriate server. 
\item Each pair $(i, j)$ of servers exchanges not only its own requests but the requests of all servers that relayed their requests either to $i$ or to $j$. Since different servers may have different 
communication delays to $i$ and $j$ the local balancing requires more care (Algorithms~\ref{alg:exchangingLoads}~and ~\ref{alg:distributedOptimal}).
\end{inparaenum}

Algorithm~\ref{alg:distributedOptimal} has the following properties: 
\begin{inparaenum}[(i)] 
\item The size of the input data is $O(m)$ for each server---communication latencies from a server to all other servers (and not for all pairs of servers). It is easy to measure these pair-wise latencies (Section~\ref{sec::introduction}). The algorithm is also applicable to the case when we allow the server to relay its requests only to the certain subset of servers (we set the latencies to the servers outside of this subset to infinity).
\item A single optimization step requires only two servers to be available (thus, it is very robust to failures).
\item Any algorithm that in a single step involves only two servers cannot perform better (Theorem~\ref{lemma::delegationOptimality}).
\item The algorithm does not require any requests to be unnecessarily delegated -- once the relay fractions are calculated the requests are sent over the network.
\item In each step of the algorithm we are able to estimate the distance between the current solution and the optimal one (Proposition~\ref{lemma::convergence}).
\end{inparaenum}

\subsection{Optimality}

The following theorem shows the optimality of Algorithm~\ref{alg:exchangingLoads}.

\begin{theorem}\label{lemma::delegationOptimality}
After execution of Algorithm~\ref{alg:exchangingLoads} for the pair of servers $i$ and $j$, $\sum C_{i}$ cannot be further improved by sending the load of any servers between $i$ and $j$ (by adjusting $r_{ki}$ and $r_{kj}$ for any $k$).
\end{theorem}
\begin{proof}
For the sake of simplicity of the presentation we prove that after performing Algorithm~\ref{alg:exchangingLoads}, for any single server $k$ we cannot improve the processing time $\sum{C_i}$ by moving any requests of $k$ from $i$ to $j$ or from $j$ to $i$. Similarly it can be proven that we cannot improve $\sum{C_i}$ by moving the requests of any \emph{set} of the servers from $i$ to $j$ or from $j$ to $i$.

Let us consider the total processing time function $h_i(l) = lf_{i}(l)$. Since $f_{i}$ is non-decreasing and convex, $h_i$ is convex. Indeed if $l > 0$, then:
\begin{align*}
h_i''(l) = (f_{i}(l) + lf_{i}'(l))' = 2f_{i}'(l) + lf_{i}''(l) > 0 \textrm{.}
\end{align*}
Now, let $l$ be the total load on the servers $i$ and $j$, $l = l_i + l_j$. Let us consider the function $P(\Delta r)$ describing the contribution in $\sum{C_i}$ of servers $i$ and $j$ as a function of load $\Delta r$ processed on the server $j$ (excluding communication): 
\begin{align*}
P(\Delta r) &= (l - \Delta r)f_i(l - \Delta r) +  \Delta rf_j(\Delta r)  \\
              &= h_i(l - \Delta r) + h_j(\Delta r) \textrm{.}
\end{align*}
The function $P$ is convex as well. Indeed:
\begin{align*}
P''(\Delta r) = h_i''(l - \Delta r) + h_j''(\Delta r) > 0 \textrm{.}
\end{align*}

Now, we show that after the second loop (Algorithm~\ref{alg:exchangingLoads}, lines~\ref{algline::secLoopStart}-\ref{algline::secLoopEnd}) transferring any load from $i$ to $j$, would not further decrease the total completion time $\sum C_{i}$. 
For the sake of contradiction let us assume that for some server $k$ after the second loop some additional requests of $k$ should be transferred from $i$ to $j$.
The second loop considers the servers in some particular order and in each iteration moves some load (possibly of size equal to 0) from $i$ to $j$.
Let $I_k$ be the iteration of the second loop in which the algorithm considers the requests owned by $k$ and tries to move some of them from $i$ to $j$.
Let $l - \Delta r_1$ and $l - \Delta r_2$ be the loads on the server $i$ immediately after $I_k$ and after the last iteration of the second loop, respectively.
As no request is moved back from $j$ to $i$, $\Delta r_2 \geq \Delta r_1$.
We will use a function $P_k$:
\begin{align*}
P_k(\Delta r) = P(\Delta r_1 + \Delta r) - \Delta r c_{ki} + \Delta r c_{kj} \textrm{.}
\end{align*}
The function $P_k(\Delta r)$ returns the total processing time of $i$ and $j$ assuming the server $i$ after iteration $I_k$ sent additional $\Delta r$ more requests of $k$ to $j$ (including the communication delay of these extra $\Delta r$ requests).

Immediately after iteration $I_k$ the algorithm could not improve the processing time of the requests by moving some requests owned by $k$ from $i$ to $j$. This is the consequence of one of two facts. Either all the requests of $k$ are already on $j$, and so there are no requests of $k$ to be moved (but in such case we know that when the whole loop is finished there are still no such requests, and thus we get a contradiction). Alternatively, the function $P_k$ is increasing for some interval $[0, \epsilon]$ ($\epsilon > 0$). But then we infer that the function:
\begin{align*}
Q_k(\Delta r) = P(\Delta r_2 + \Delta r) - \Delta r c_{ki} + \Delta r c_{kj} \textrm{,}
\end{align*}
is also increasing on $[0, \epsilon]$. Indeed:
\begin{align*}
Q_k'(\Delta r) = P'(\Delta r_2 + \Delta r) - c_{ki} + c_{kj} \geq P'(\Delta r_1 + \Delta r) - c_{ki} + c_{kj} = P_k'(\Delta r)\textrm{,}
\end{align*}
Since $Q_k$ is convex (because $P$ is convex) we get that $Q_k$ is increasing not only on $[0, \epsilon]$, but also for any positive $\Delta r$. Thus, it is not possible to improve the total completion time by sending the requests of $k$ from $i$ to $j$ after the whole loop is finished. This gives a contradiction.

Second, we will show that when the algorithm finishes no requests should be transferred back from $j$ to $i$ either.
Again, for the sake of contradiction let us assume that for some server $k$ after the second loop (Algorithm~\ref{alg:exchangingLoads}, lines~\ref{algline::secLoopStart}-\ref{algline::secLoopEnd}) some requests of $k$ should be transferred back from $j$ to $i$.
Let $I_k$ be the iteration of the second in which the algorithm considers the requests owned by $k$.
Let us take the last iteration $I_{s_{\mathit{last}}}$ of the second loop in which the requests of some server $s_{\mathit{last}}$ were transferred from $i$ to $j$.
Let $l - \Delta r_3$ be the load on $i$ after $I_{s_{\mathit{last}}}$. 
After $I_{s_{\mathit{last}}}$ no requests of $s_{\mathit{last}}$ should be transferred back from $j$ to $i$ ($\mathrm{argmin}$ in line~\ref{algline::argmin}). Thus, for some $\epsilon > 0$ the function $R_k$:
\begin{align*}
R_k(\Delta r) = P(\Delta r_3 - \Delta r) + \Delta r c_{s_{\mathit{last}} i} - \Delta r c_{s_{\mathit{last}} j}
\end{align*}
is increasing on $[0, \epsilon]$. Since the servers are ordered by decreasing latency differences  $(c_{ki} - c_{kj})$ (increasing latency differences  $(c_{kj} - c_{ki})$), we get $c_{s_{\mathit{last}} i} - c_{s_{\mathit{last}} j} \leq c_{ki} - c_{kj}$, and so that the function: 
\begin{align*}
S_k(\Delta r) = P(\Delta r_3 - \Delta r) + \Delta r c_{k i} - \Delta r c_{k j}
\end{align*}
is also increasing on $[0, \epsilon]$. Since $S_k$ is convex we see that it is increasing or any positive $\Delta r$, and thus we get the contradiction. This completes the proof.
\end{proof}

\subsection{Convergence}\label{sec:convergence}
The following analysis bounds the error of the distributed algorithm as a function of the servers' loads. When running the algorithm, this result can be used to assess whether it is still profitable to continue. As the corollary of our analysis we will show the convergence of the distributed algorithm. 

In proofs, we will use an \emph{error graph} that quantifies the difference of loads between the current and the optimal solution.

\begin{definition}[Error graph]
Let $\rho$ be the snapshot (the current solution) at some moment of execution of the distributed algorithm. 
Let $\rho^{*}$ be the optimal solution (if there are multiple optimal solutions with the same $\sum{C_i}$, $\rho^{*}$ is the closest solution to $\rho$ in the Manhattan metric).
$(P, \Delta \rho)$ is a weighted, directed \emph{error graph} with multiple edges. The vertices in the error graph correspond to the servers; $\Delta \rho[i][j][k]$ is a weight of the edge $i \rightarrow j$ with a label $k$. The weight indicates the number of requests owned by $k$ that should be executed on $j$ instead of $i$ in order to reach $\rho^{*}$ from $\rho$.
\end{definition}

The error graphs are not unique. For instance, to move $x$ requests owned by $k$ from $i$ to $j$ we can move them directly, or through some other server $\ell$. In our analysis we will assume that the total weight of the edges in the error graph $\sum_{i, j, k}\Delta \rho[i][j][k]$ is minimal, that is that there is no $i, j, k$, and $\ell$, such that $\Delta \rho[i][\ell][k] > 0$ and $\Delta \rho[\ell][j][k] > 0$.

Let $succ(i) = \{j: \exists_k \Delta \rho[i][j][k] > 0\}$ denote the set of (immediate) successors of server $i$ in the error graph; $prec(i)= \{j: \exists_k \Delta \rho[j][i][k] > 0\}$ denotes the set of (immediate) predecessors of $i$.

We will also use a notion of \emph{negative cycle}: a sequence of servers in the error graph that essentially redirect some of their requests to one another.
\begin{definition}[Negative cycle]
In the error graph, a \emph{negative cycle} is a sequence of servers $i_{1}, i_{2}, \ldots, i_{n}$ and labels $k_{1}, k_{2}, \ldots, k_{n}$ such that:
\begin{enumerate}
\item $i_{1} = i_{n}$; (the sequence is a cycle)
\item $\forall_{j \in \{1,\ldots n-1\}}\; \Delta \rho[i_{j}][i_{j+1}][k_j] > 0$; (for each pair there is an edge in the error graph)
\item $\sum_{j=1}^{n-1} c_{k_{j}i_{j+1}} < \sum_{j=1}^{n-1} c_{k_{j}i_{j}}$ (the transfer in the circle $i_{j} \xrightarrow{k_j} i_{j+1}$ decreases communication delay).
\end{enumerate}
\end{definition}
A current solution that results in an error graph without negative cycles has smaller processing time: after dismantling a negative cycle, loads on servers remain the same, but the communication time is reduced. Thus, if the current solution has an optimal network flow, then there are no negative cycles in the error graph. 

Analogously we define \emph{positive cycles}. The only difference is that instead of the third inequality we require $\sum_{j=1}^{n-1} c_{k_{j}i_{j+1}} \geq \sum_{j=1}^{n-1} c_{k_{j}i_{j}}$. Thus, when an error graph has a positive cycle, the current solution is better than if the cycle would be dismantled.

We start by bounding the load imbalance when there are no negative cycles.

\begin{lemma}\label{lemma::convergence}
Let $\mathit{impr}_{pq}$ be the improvement of the total processing time $\sum C_{i}$ after balancing servers $p$ and $q$ by Algorithm~\ref{alg:exchangingLoads}. 
Let $l_i$ be the load of a server $i$ in the current state; and $l^{*}_i$ be the optimal load. 
If the error graph $\Delta \rho$ has no negative cycles, then for every positive $\epsilon$ the following estimation holds:
\begin{align*}
f_i(l_i) -  f_i(l^{*}_{i}) \leq \frac{6U_1 + 3\lmax U_2}{\epsilon} \max_{pq} \mathit{impr}_{pq} + m\epsilon \text{.}
\end{align*}
\end{lemma}

\begin{proof}
First we show that there is no cycle (positive nor negative) in the error graph. By contradiction let us assume that there is a cycle: $i_1, \ldots, i_{n-1}, i_{n}$ (with $i_1=i_n$) with labels $k_{1}, k_{2}, \ldots, k_{n}$.
Because, we assumed the error graph has no negative cycle, we have: $\sum_{j=1}^{n-1} ( c_{k_{j}i_{j+1}} - c_{k_{j}i_{j}} ) \geq 0$. 
Now, let $\Delta \rho_{\mathit{min}} = \min_{j \in \{1,\ldots, n-1\}}( \rho [i_{j}][i_{j+1}][k_{j}])$ be the minimal load on the cycle. If we reduce the number of requests sent on each edge of the cycle:
\begin{align*}
\Delta \rho [i_{j}][i_{j+1}][k_{j}] := \Delta \rho [i_{j}][i_{j+1}][k_{j}] - \Delta \rho_{\mathit{min}}
\end{align*}
then the load of the servers $i_{j}, j \in \{1, \ldots, n-1 \}$ will not change.
Additionally, the latencies decrease by $\rho_{\mathit{min}}\left(\sum_{j=1}^{n-1} c_{k_{j}i_{j+1}} - c_{k_{j}i_{j}}) \right)$ which is at least equal to 0. Thus, we get a new optimal solution which is closer to $\rho$ in Manhattan metric, which contradicts that $\rho^{*}$ is optimal.

In the remaining part of the proof, we show how to bound the difference $|f_i(l_{i}) - f_i(l^{*}_{i})|$.
Consider a server $i$ for which $l_{i} > l^{*}_{i}$, and a server $j \in succ(i)$. 
% Now we define $\Delta r^{\epsilon}_{ij}$ in the following way.
We define as $\Delta r^{\epsilon}_{ij}$ the load that in the current state $\rho$ should be transferred between $i$ and $j$ so that after this transfer, moving any $\Delta r$ more load owned by any $k$ between $i$ to $j$ would be either impossible or would not improve $\sum{C_i}$ by more than $\epsilon \Delta r$. 
Intuitively, after moving $\Delta r^{\epsilon}_{ij}$, we won't be able to further ``significantly'' reduce $\sum{C_i}$: further reductions depend on the moved load ($\Delta r$), but the rate of the improvement is lower than $\epsilon$. This move resembles Algorithm~\ref{alg:exchangingLoads}: i.e., Algorithm~\ref{alg:exchangingLoads} moves $\Delta r^{\epsilon}_{ij}$ for $\epsilon=0$.

Let $\tilde{\rho}$ denote a state obtained from $\rho$ when $i$ moves to $j$ exactly $\Delta r^{\epsilon}_{ij}$ requests.
Let $\tilde{l_i}$ and $\tilde{l_j}$ denote the loads of the servers $i$ and $j$ in $\tilde{\rho}$, respectively.
We define $H_k(\Delta r)$ as the change of $\sum{C_i}$ resulting from moving additional $\Delta r$ requests produced by $k$ from $i$ to $j$:
\begin{align*}
H_k(\Delta r) = h_i(\tilde{l_i} - \Delta r) + h_j(\tilde{l_j} + \Delta r) - \Delta r c_{ki} + \Delta r c_{kj} \textrm{.}
\end{align*}

State $\tilde{\rho}$ satisfies one of the following conditions for each $k$ (consider $\Delta r$ as a small, positive number):
\begin{enumerate}
\item The servers $i$ and $j$ are $\epsilon$-balanced, thus moving $\Delta r$ requests from $i$ to $j$ would not reduce $\sum C_i$ by more than $\epsilon \Delta r$. In such case we can bound the derivative of $H_k$:
\begin{align}
|H'_k(0)| \leq \epsilon \label{ineq:convergence1} \textrm{,}
\end{align}
\item Or, moving $\Delta r$ requests from $i$ to $j$ would decrease $\sum C_i$ by more than $\epsilon \Delta r$, but there are no requests of $k$ on $i$:
\begin{align}
H'_k(0) < -\epsilon \;\;\; \mathrm{and}\;\;\; \tilde{r_{ki}} = 0  \label{ineq:convergence2} \textrm{,}
\end{align}
% In other words, although loads of $i$ and $j$ are in such relation that sending some sufficiently small number of requests $\Delta r$ from $i$ to $j$ would decrease the total completion time by at least $\Delta r \epsilon$ ($H'_k(0) < -\epsilon$), there are no requests of $k$ on $i$ ($r'_{ki} = 0$).
\item Or, moving $\Delta r$ requests back from $j$ to $i$ would decrease $\sum C_i$ by more than $\epsilon \Delta r$, but there are no requests of $k$ to be moved back:
\begin{align}
H'_k(0) > \epsilon \;\;\; \mathrm{and}\;\;\; \tilde{r_{kj}} = 0  \label{ineq:convergence3} \textrm{.}
\end{align}
\end{enumerate}

In the optimal solution, for any $k$, no $k$'s requests should be moved between $i$ and $j$. We define $G_k(\Delta r)$ similarly to $H_k$, but for the optimal loads:
\begin{align}
G_k(\Delta r) = h_i(l^{*}_i - \Delta r) + h_j(l^{*}_j + \Delta r) - \Delta r c_{ki} + \Delta r c_{kj} \textrm{.}
\end{align}
By the same reasoning, at least one of the three following inequalities holds:
\begin{alignat}{2}
& G'_k(0) = 0 \text{,} \ & \ & \text{or} \label{ineq:convergence4} \\
& G'_k(0) < 0 \;\;\; \mathrm{and}\;\;\; r^{*}_{ki} = 0 \text{,} \ & \ & \text{or} \label{ineq:convergence5} \\
& G'_k(0) > 0 \;\;\; \mathrm{and}\;\;\; r^{*}_{kj} = 0 \text{.} \ & \ &  \label{ineq:convergence6}
\end{alignat}

% Now we will use the above observations for deriving further estimations. 
% For each third server $k$, one of the two following cases happens. In the error graph, either (i) $i$ does not transfer to $j$ any requests of $k$ ($\Delta \rho[i][j][k] = 0$), or (ii) $i$ transfers some requests of $k$ to $j$ ($\Delta \rho[i][j][k] > 0$).
% In (ii), $r^{*}_{kj} > 0$ and as Inequality~\ref{ineq:convergence6} cannot hold, 
% % be true. In such case Inequality~\ref{ineq:convergence4},~or~Inequality~\ref{ineq:convergence5} must hold, and so
% and we get from Inequalities~\ref{ineq:convergence4} and~\ref{ineq:convergence5} that $G'_k(0) \leq 0$.
% Additionally, we consider two sub-cases. Either (ii.a) $H'_k(0) \geq -\epsilon$, or (ii.b) $H'_k(0) < -\epsilon$. If $H'_k(0) < -\epsilon$ (ii.b), then Inequality~\ref{ineq:convergence2} holds; from which we get $r'_{ki} = 0$.
% To summarize cases (i) and (ii): either $i$ in the error graph does not transfer to $j$ any requests of $k$ (case (i)), or $r'_{ki} = 0$ (ii.b), or $G'_k(0) \leq 0$ and $H'_k(0) \geq -\epsilon$ (ii.a). 

We consider two 
% further 
cases on the sum of weights between $i$ and $j$ in the error graph.
Either (1) in the error graph, $i$ sends to $j$ at most $\Delta r^{\epsilon}_{ij}$ requests ($\sum_k \Delta \rho[i][j][k] \leq \Delta r^{\epsilon}_{ij}$); or (2) 
$\sum_k \Delta \rho[i][j][k] > \Delta r^{\epsilon}_{ij}$. We further analyze (2). Since $\Delta r^{\epsilon}_{ij}$ is the total load transferred from $i$ to $j$ in $\rho$ to get $\tilde{\rho}$, there must exist a $k$ such that $\rho[i][j][k] > \Delta r^{\epsilon}_{ij}(k)$ (from $i$ to $j$ more $k$'s requests are moved in the error graph than in $\rho$ to get $\tilde{\rho}$). 
% $\rho[i][j][k] > 0$
We show that $\tilde{r_{ki}} > 0$ by contradiction. 
If $\tilde{r_{ki}} = 0$ (in $\tilde{\rho}$, 
no $k$'s requests are processed on $i$), 
then $r_{ki} = \Delta r^{\epsilon}_{ij}(k)$ 
(all $k$'s requests were moved to $j$ in $\rho$ to get $\tilde{\rho}$).
As $\rho[i][j][k] \leq r_{ki}$ (the error graph does not transfer more requests than available),
$\rho[i][j][k] \leq \Delta r^{\epsilon}_{ij}(k)$, which contradicts $\rho[i][j][k] > \Delta r^{\epsilon}_{ij}(k)$.
As $\tilde{r_{ki}} > 0$,
Ineq.~\ref{ineq:convergence1} or~\ref{ineq:convergence3} holds (Ineq.~\ref{ineq:convergence2} does not hold),
thus $H_k'(0) \geq -\epsilon$.
As $\rho[i][j][k] > \Delta r^{\epsilon}_{ij}(k)$,
$\rho[i][j][k] > 0$,
thus, $r^*_{kj} > 0$,
so Ineq.~\ref{ineq:convergence4} or~\ref{ineq:convergence5} holds (Ineq.~\ref{ineq:convergence6} does not hold),
so $G'_k(0) \leq 0$.

% thus: either $i$ in the error graph sends to $j$ at most $\Delta r^{\epsilon}_{ij}$ (and so $|l_{i} - l^{*}_{i}| \leq \Delta r^{\epsilon}_{ij}$), or
\begin{align*}
G'_k(0) \leq 0 &\Leftrightarrow -h'_i(l^{*}_{i}) + h'_j(l^{*}_{j}) - c_{ki} + c_{kj}  \leq 0 \\
H'_k(0) \geq -\epsilon &\Leftrightarrow -h'_i(\tilde{l_{i}}) +  h'_j(\tilde{l_{j}}) - c_{ki} + c_{kj} \geq -\epsilon
\end{align*}
Combining the above inequalities,
\begin{align*}
-h'_i(\tilde{l_{i}}) +  h'_j(\tilde{l_{j}}) + \epsilon \geq -h'_i(l^{*}_{i}) + h'_j(l^{*}_{j})
\end{align*}
Or equivalently:
\begin{align}
h'_i(\tilde{l_{i}}) - h'_i(l^{*}_{i}) \leq h'_j(\tilde{l_{j}}) - h'_j(l^{*}_{j}) + \epsilon \label{ineq:relation-expanded2}
\end{align}

We can further expand the above inequality for $j$ and its successors (and each expansion is applied on state $\rho$), and so on towards the end of the error graph (we proved there are no cycles), until we get that for some $p$ and its successor $q$ the condition of the case (2) does not hold, so (1) must hold ($\sum_k \Delta \rho[p][q][k] \leq \Delta r^{\epsilon}_{pq}$, or equivalently $|l_{p} - l^{*}_{p}| \leq \Delta r^{\epsilon}_{p q}$). 
Analogously to $\tilde{\rho}$ we define $\tilde{\tilde{\rho}}$ as the state in which $p$ moves to $q$ load $\Delta r^{\epsilon}_{pq}$.
Thus, we have:
\begin{align}
 h'_i(\tilde{l_{i}}) - h'_i(l^{*}_{i}) \leq h'_p(\tilde{\tilde{l_{p}}}) - h'_p(l^{*}_{p}) + m\epsilon & \; \; \text{(Ineq. \ref{ineq:relation-expanded2} expanded for at most $m$ successors)} \label{ineq:derivativesComp0} \\
|l_{p} - l^{*}_{p}| \leq \Delta r^{\epsilon}_{pq} & \; \; \text{(condition (1))}\label{ineq:loadDiff}
\end{align}
From the definition of $\tilde{\tilde{\rho}}$ we have:
\begin{align}
& |\tilde{\tilde{l_{p}}} - l_{p}| \leq \Delta r^{\epsilon}_{pq} \label{ineq:loadDiff1}
\end{align}
Combining Inequalities~\ref{ineq:loadDiff}~and~\ref{ineq:loadDiff1} we get:
\begin{align}
& |\tilde{\tilde{l_{p}}} - l^{*}_{p}| \leq 2\Delta r^{\epsilon}_{pq} \label{ineq:loadDiff2}
\end{align}
We bound the second derivative of $h''_p$:
\begin{align}
h_{p}''(l) = (f_p(l) + l f_p'(l))' = 2f_p'(l) + l f_p''(l) \leq 2U_1 + \lmax U_2 = U_3 \label{ineq:loadDiff3}\textrm{.}
\end{align}
With the above observations, and using the bound from Inequality~\ref{eq:der-upperbound} we get: 
\begin{align*}
h_i'(l_{i}) - h_i'(l^{*}_{i})  & \leq h_i'(\tilde{l_{i}}) - h_i'(l^{*}_{i}) + U_3 |\tilde{l_{i}} - l_{i}| && \text{Ineq.~\ref{eq:der-upperbound}} \\
                             & \leq h_i'(\tilde{l_{i}}) - h_i'(l^{*}_{i}) + U_3 \Delta r^{\epsilon}_{ij} && \text{condition (1)} \\
                             & \leq h_p'(\tilde{\tilde{l_{p}}}) - h_p'(l^{*}_{p}) + U_3 \Delta r^{\epsilon}_{ij}  + m\epsilon && \text{Ineq.~\ref{ineq:derivativesComp0}}\\
                             & \leq U_3 |\tilde{\tilde{l_{p}}} - l^{*}_{p}| + U_3 \Delta r^{\epsilon}_{ij}  + m\epsilon && \text{Ineq.~\ref{eq:der-upperbound}} \\
                             & \leq 2 U_3 \Delta r^{\epsilon}_{pq} + U_3 \Delta r^{\epsilon}_{ij}  + m\epsilon && \text{Ineq.~\ref{ineq:loadDiff2}}\\
                             & \leq (6U_1 + 3\lmax U_2) \max_{pq} \Delta r^{\epsilon}_{pq} + m\epsilon && \text{Ineq.~\ref{ineq:loadDiff3}.}
\end{align*}
As $h_i'(l_{i}) = f_i(l_i) + l_if_i'(l_i)$, from $l_i \geq l^{*}_{i}$, we get that:
\begin{align*}
f_i(l_i) -  f_i(l^{*}_{i}) & \leq h_i'(l_{i}) - h_i'(l^{*}_{i}) \\
& \leq (6U_1 + 3\lmax U_2) \max_{pq} \Delta r^{\epsilon}_{pq} + m\epsilon \text{.}
\end{align*}

We can get the same results for the server $i$ such that $l_{i} \leq l^{*}_{i}$, by expanding the inequalities towards the predecessors of $i$.

We now relate $\Delta r^{\epsilon}_{pq}$ with $\mathit{impr}_{pq}$, the result of Algorithm~\ref{alg:exchangingLoads}. Algorithm~\ref{alg:exchangingLoads} stops when no further improvement is possible, thus the load moved by the algorithm is at least $\Delta r^{\epsilon}_{pq}$ for any $\epsilon$. By definition of $\Delta r^{\epsilon}_{pq}$, moving $\Delta r^{\epsilon}_{pq}$ load improves $\sum{C_i}$ by at least $\epsilon \Delta r^{\epsilon}_{pq}$. Thus, 
$\mathit{impr}_{ij} \geq \epsilon \Delta r^{\epsilon}_{ij}$. 
This completes the proof.
\end{proof}

As the result we get the following corollary.

\begin{corollary}
If the network flow in the current solution $\rho$ is optimal, then the absolute error $e$ is bounded:
\begin{align*}
e \leq \ltot \frac{6U_1 + 3\lmax U_2}{\epsilon} \max_{pq} \mathit{impr}_{pq} + m\ltot\epsilon
\end{align*}
\end{corollary}
\begin{proof}
The value $f_i(l_i)$ denotes the average processing time of a request on the $i$-th server. For every server $i$ the average processing time of every request on $i$ in $\rho$ is by at most $\frac{6U_1 + 3\lmax U_2}{\epsilon} \max_{pq} \mathit{impr}_{pq} + m\epsilon$ greater than in $\rho^{*}$. Thus, since there are $\ltot$ requests in total, we get the thesis.
\end{proof}

We can use Lemma~\ref{lemma::convergence} directly to estimate the error during the optimization if we run a distributed negative cycle removal algorithm (e.g.~\cite{minimumCirculation, auctionBasedMinCostFlow}).
However, this result is even more powerful  when applied together with the lemmas below, as it will allow to bound the speed of the convergence of the algorithm (even without additional protocols optimizing the network flow).
Now, we show how to bound the impact of the negative cycles.

\begin{lemma}\label{lemma::negativeCyclesElimination}
For every $\epsilon > 0$, removing the negative cycles improves the total processing time $\sum C_{i}$ of the solution by at most:
\begin{align*}
\epsilon \ltot + 2m \sum_{ij} \impr_{ij} + \frac{16U_1 + 8\lmax U_2}{\epsilon} \max_{ij} \impr_{ij} \ltot \text{.}
\end{align*}
\end{lemma}

\begin{proof}
In the analysis we will use a function $F_{i, j, k}$ defined as $H_k$ in the proof of Lemma~\ref{lemma::convergence} (here, we will use indices $i,j$) 
\begin{align*}
F_{i, j, k}(\Delta r) &= h_{i}(l_i - \Delta r) + h_{j}(l_j + \Delta r) - \Delta r c_{ki} + \Delta r c_{kj} \\
                      &= (l_i - \Delta r)f_{i}(l_i - \Delta r) + (l_j + \Delta r)f_{j}(l_j + \Delta r) - \Delta r c_{ki} + \Delta r c_{kj} \text{.}
\end{align*}

We will analyze a procedure that removes negative cycles one by one. 

First, we prove that we can remove all the negative cycles by only considering the cycles that satisfy the \emph{one-way transfers} property: if in a cycle there is a transfer $i \to j$, then in no cycle there is a transfer $j \to i$.
Indeed, consider the set of all cycles. We do the following procedure:
\begin{enumerate}
\item If there are two negative cycles $C$ and $\tilde{C}$ with a common edge $(i, j)$, such that in the first cycle load $l$ is transferred from $i$ to $j$ and in the second load $\tilde{l} \leq l$ is
transferred back from $j$ to $i$ 
then we split the first cycle $C$ into two cycles $C_1$ and $C_2$ such that $C_1$ transfers $\tilde{l}$ and $C_2$ transfers $l - \tilde{l}$. Next, we merge $C_1$ and $\tilde{C}$ into one negative cycle that does not contain edge $(i, j)$. 
\item If a single cycle transfers load first from $i$ to $j$, and then from $j$ to $i$ then we break the cycle at the edge $i \leftrightarrow j$ into two cycles (without edges between $i$ and $j$).  \label{enum:cycle-break2}
\end{enumerate}
Let us note that each of the two above steps does not increase the total load transferred on the cycles.
For a given error graph, there are many possible ways (sets of cycles) to express a non-optimal flow as a sum of negative cycles. We will consider a set of cycles with the smallest total load (sum of loads transferred over all edges of all cycles).
% Thus, we can consider the sequence of the cycles to remove that totally transfer the smallest possible total. 
% If we consider such a sequence, then every request will be transferred only once. 
In this set, we will remove individual cycles sequentially in an arbitrary order.

In this sequence of cycles with the smallest load, every request is transferred at most once.
Indeed, if this is not the case, a request was transferred through adjacent edges $e_1$ and $e_2$. Thus, among the cycles we consider there are two negative cycles, such that the first one contains the edge $e_1 = (i, j)$ and the second one contains $e_2 = (j, \ell)$. 
(a single request cannot be transferred $i \to j \to \ell$ in a single cycle, because by sending $i \to \ell$ we would get a cycle with a smaller load).
Also, between $i$ and $j$ and between $j$ and $\ell$ requests of the same server $k$ are transferred. Let $e_3 = (\ell, p)$ be the edge adjacent to $e_2$ in the second cycle.
If in the first cycle we send requests from $i$ to $\ell$ and in the second from $j$ to $p$, we would obtain two cycles that transfer an equivalent load (each server has the same requests) and have
smaller total transfer, a contradiction.
%As a consequence the total load transferred on the considered cycles through a single edge is at most equal to $l_{\mathit{tot}}$. 

Let us consider a state $\rho$ with a negative cycle $c$, that is the sequence of servers $i_{1}, i_{2}, \ldots, i_{n}$ and labels $k_{1}, k_{2}, \ldots, k_{n}$. Let us assume that in a negative cycle $c$ the load $\Delta r$ is carried on. After removing the cycle $c$, $\sum{C_i}$ is improved by $I_{c}$:
\begin{align*}
0 < I_{c} &= \sum_{j} \Delta r \left(c_{k_j i_{j}} - c_{k_j i_{j+1}}\right) \\
          &= \Delta r \sum_{j} \left(h'_{i_j}(l_{i_j}) - h'_{i_{j+1}}(l_{i_{j+1}}) + c_{k_j i_{j}} - c_{k_j i_{j+1}}\right) && \text{the sequence is a cycle}\\
          &= \Delta r \sum_{j} -F'_{i_j, i_{j+1}, k_j}(0) \text{.}
\end{align*}
Let us distribute among the edges the cost of all negative cycles (the cost, i.e., the increase in $\sum C_i$ resulting from the cycles). For removing a single cycle with load $\Delta r$, from the above inequality we assign the cost $-F'_{i, j, k}(0)\Delta r$ to the labeled edge $i \xrightarrow{k} j$. As every request is sent over a single edge at most once, the total cost assigned to the labeled edge $i \xrightarrow{k} j$ will be at most $-F'_{i, j, k}(0) r_{ki}$. 

In the further part of the proof we will estimate $-F'_{i, j, k}(0)$.
First, we bound the second derivative of $F$ as in Eq.~\ref{eq:der2-upperbounded} in the proof of Theorem~\ref{thm:approxQuality}:
\begin{align*}
|F''_{i, j, k}(\Delta r)| \leq 4U_1 + 2\lmax U_2 \textrm{.}
\end{align*}
Next, we consider two cases: (1) $F'_{i, j, k}(0) \geq -\epsilon$, and (2) $F'_{i, j, k}(0) < -\epsilon$. If (1) is the case then the total cost associated with $i \xrightarrow{k} j$ is at most $\epsilon r_{ki}$. We further analyze (2). Let us take $\Delta r_0 = \min(r_{ki}, \frac{\epsilon}{8U_1 + 4\lmax U_2})$. From Inequality~\ref{eq:der-upperbound} we get that: 
\begin{align*}
F'_{i, j, k}(\Delta r_0) &\leq F'_{i, j, k}(0) + \Delta r_0(4U_1 + 2\lmax U_2) \\
 & \leq F'_{i, j, k}(0) +\frac{\epsilon}{8U_1 + 4\lmax U_2}(4U_1 + 2\lmax U_2)  \\
& \leq F'_{i,j,k}(0) - \frac{F'_{i, j, k}(0)}{2} \leq \frac{F'_{i, j, k}(0)}{2} \textrm{.}
\end{align*}
Thus (again from Ineq.~\ref{eq:der-upperbound} but applied for $F$) we get that:
\begin{align*}
F_{i, j, k}(0) - F_{i, j, k}(\Delta r_0) \geq \Delta r_0 \frac{-F'_{i, j, k}(0)}{2} \text{.}
\end{align*}
Since $\Delta r_0 \leq r_{ki}$ (there are at least $\Delta r_0$ requests of $k$ on $i$), we infer that Algorithm~\ref{alg:exchangingLoads} would achieve improvement $\impr_{ij}$ lower-bounded by:
\begin{align*}
\impr_{ij} \geq F_{i, j, k}(0) - F_{i, j, k}(\Delta r_0) \geq  -F'_{i, j, k}(0)\frac{\Delta r_0}{2} \text{.}
\end{align*}
Further, we consider two sub-cases. If (2a) $\Delta r_0 = r_{ki}$ then the total cost associated with $i \xrightarrow{k} j$ is at most $2\impr_{ij}$. (2b) Otherwise ($\Delta r_0 = \frac{\epsilon}{8U_1 + 4\lmax U_2}$), we have $\impr_{ij} \frac{16U_1 + 8\lmax U_2}{\epsilon} \geq -F'_{i, j, k}(0)$.
Since $\sum_{i,k}r_{ki} = \ltot$, we get that the total cost associated with all edges is at most:
\begin{align*}
\epsilon \ltot + && \text{condition (1)} \\
2m \sum_{ij} \impr_{ij} + && \text{condition (2a)} \\
\frac{16U_1 + 8\lmax U_2}{\epsilon} \max_{ij} \impr_{ij} \ltot \text{.} && \text{condition (2b)}
\end{align*}
Thus, we get the thesis. 
\end{proof}

Finally we get the following estimation.

\begin{theorem}\label{thm::finalEstimation}
Let $\mathit{impr}_{ij}$ be the improvement of the total processing time $\sum C_{i}$ after balancing servers $i$ and $j$ through Algorithm~\ref{alg:exchangingLoads}. Let $e$ be the absolute error in $\sum C_{i}$ (the difference between $\sum C_{i}$ in the current and the optimal state). For every $\epsilon > 0$, we have:
\begin{align*}
e \leq 2m \sum_{ij} \impr_{ij} + \max_{ij} \impr_{ij} \frac{22U_1 + 11\lmax U_2}{\epsilon} \ltot + (m+1)\ltot\epsilon \textrm{.}
\end{align*}
\end{theorem}
\begin{proof}
The error coming from the negative cycles is bounded by Lemma~\ref{lemma::negativeCyclesElimination} by:
\begin{align}
\epsilon \ltot + 2m \sum_{ij} \impr_{ij} + \frac{16U_1 + 8\lmax U_2}{\epsilon} \max_{ij} \impr_{ij} \ltot \textrm{.}
\end{align}
The error coming from the processing times is, according to Lemma~\ref{lemma::convergence} bounded by:
\begin{align*}
\ltot \frac{6U_1 + 3\lmax U_2}{\epsilon} \max_{ij} \mathit{impr}_{ij} + m\ltot\epsilon
\end{align*}
The sum of the above errors leads to the thesis.
\end{proof}

And the following theorem.

\begin{theorem}\label{thm::convergence}
Let $e_i$ and $e_d$ be the initial and the desired absolute errors. The distributed algorithm reaches the $e_d$ in expected time complexity:
\begin{align*}
O\left(\frac{\ltot^2 (U_1 +\lmax U_2)e_i m^3}{e_d^2}\right) \textrm{.}
\end{align*}
\end{theorem}
\begin{proof}
In the estimation from Theorem~\ref{thm::finalEstimation} we set $\epsilon = \frac{e_d}{2(m+1)\ltot}$ and relax the upper bound by replacing $\max_{i,j} \mathit{impr}_{ij}$ with $\sum_{i, j} \mathit{impr}_{ij}$:
\begin{align*}
e &\leq (2m+2)\sum_{i, j} \impr_{ij}\left(1 + \frac{22U_1 + 11\lmax U_2}{e_d} \ltot^2\right) + \frac{e_d}{2} \\
  &\approx 2m\sum_{i, j} \impr_{ij}\frac{22U_1 + 11\lmax U_2}{e_d} \ltot^2 + \frac{e_d}{2} \textrm{.}
\end{align*}
Thus, either
\begin{align*}
2m\sum_{ij} \impr_{ij}\frac{22U_1 + 11\lmax U_2}{e_d} \ltot^2 \leq \frac{e_d}{2} \textrm{,}
\end{align*}
and the algorithm has already reached the error $e_d$; or in every execution step we have:
\begin{align*}
\sum_{i, j} \mathit{impr}_{ij} \geq \frac{e_d^2}{4m (22U_1 + 11\lmax U_2) \ltot^2} \text{.}
\end{align*}
The expected improvement of the distributed algorithm during every pairwise communication is $\frac{1}{m^2}\sum_{i, j} \mathit{impr}_{ij}$, and thus it is lower bounded by:
\begin{align*}
\frac{e_d^2}{4m^3 (22U_1 + 11\lmax U_2) \ltot^2} \textrm{.}
\end{align*}
Thus, after, in expectation, $O(\frac{\ltot^2 (U_1 +\lmax U_2)e_i m^3}{e_d^2})$ steps the initial error drops to 0.
This completes the proof.
\end{proof}

For the relative errors $e_{i,r} = \frac{e_i}{\ltot}$, and $e_{d,r} = \frac{e_d}{\ltot}$, the complexity of the algorithm is equal to $O(\frac{\ltot (U_1 +\lmax U_2)e_{i, r} m^3}{e_{d, r}^2})$.

\section{Conclusions}

In this paper we considered the problem of balancing the load between geographically distributed servers. In this problem the completion time of a request is sum of the 
communication latency needed to send the request to a server and the servers' processing time. The processing time on a server is described by a load function and depends on the total load on the server. Throughout the paper we considered a broad class of load functions with  the mild assumptions that they are convex and two times differentiable. 

We presented two algorithms---the centralized one and the distributed one. Both algorithms are any-time algorithms that continuously optimize the current solution. We shown that both algorithms converge for almost arbitrary load function. We also presented bounds on speed of their convergence that depend (apart from the standard parameters) on the bounds on the first and second derivatives of the load functions. The centralized algorithm decreases an initial relative error $e_{i, r}$ to a desired value  $e_{d, r}$ in time $O(\frac{\ltot(U_1 + \lmax U_2)e_{i, r}}{e_{d, r}^2}m^4)$. The distributed algorithm decreases $e_{i, r}$ to $e_{d, r}$ in time $O(\frac{\ltot (U_1 +\lmax U_2)e_{i, r}}{e_{d, r}^2}m^3)$. Also, for the large values of initial error $e_{i, r}$ the centralized algorithm decreases the error by half in time $O(\frac{\ltot (U_1 + \lmax U_2)}{e_{i, r}}m^5 \log m)$.

The distributed algorithm is based on the idea of gossiping. To perform a single optimization step, the algorithm requires just two servers to be available. Thus, the algorithm is robust to transient failures. It also does not require additional protocols. In some sense it is also optimal: we  proved that the local optimization step performed by this algorithm cannot be improved. Finally, at any time moment, during the execution of the distributed algorithm, we are able to assess the current error.

Experimental results were shown for a different algorithm applied to the queuing model~\cite{Liu:2011:GGL:1993744.1993767}; and for a version of our distributed algorithm specialized to the batch model~\cite{Skowron:2013:NDL:2510648.2510769}.
In our future work we plan to experimentally assess the performance of our algorithms on real workloads and several load functions, including the queuing model.

\bibliographystyle{abbrv}
\bibliography{contentDelivery}

\end{document}